\documentclass[fleqn,10pt]{wlscirep}

\usepackage[utf8]{inputenc}
\usepackage[T1]{fontenc}
\usepackage{amsmath}
\usepackage{amssymb}
\usepackage{amsthm}
\usepackage{bm}
\usepackage{tikz}
\usepackage{mathtools}
\newtheorem{theorem}{Theorem}
\DeclareMathOperator{\Tr}{tr} %
\newcommand{\dd}{{\rm d}}
\newcommand{\A}{\mathbf{D}}
\newcommand{\B}{\mathbf{B}}
\newcommand{\R}{\mathbf{R}}
\newcommand{\I}{\mathbf{I}}

\newcommand{\mT}{\intercal} 
\renewcommand{\leq}{\leqslant}
\renewcommand{\geq}{\geqslant}
\renewcommand{\le}{\leqslant}
\renewcommand{\ge}{\geqslant}

\newcommand{\EOO}[1]{\textcolor{black} {#1}}

\newcommand{\CYY}[1]{\textcolor{black} {#1}}

\title{Orientationally-averaged diffusion-attenuated magnetic resonance signal for locally-anisotropic diffusion }

\author[1,*]{Magnus Herberthson}
\author[2]{Cem Yolcu}
\author[2]{Hans Knutsson}
\author[2,4]{Carl-Fredrik Westin}
\author[2,3]{Evren \"Ozarslan}
\affil[1]{Dept.\ of Mathematics, Link\"oping University, Link\"oping, Sweden}
\affil[2]{Dept. of Biomedical Engineering, Link\"oping University, Link\"oping, Sweden}
\affil[3]{Center for Medical Image Science and Visualization, Link\"oping University, Link\"oping, Sweden}
\affil[4]{Laboratory for Mathematics in Imaging, Dept.\ of Radiology, Brigham and Women's Hospital, Harvard Medical School, Boston, MA, USA}

\affil[*]{magnus.herberthson@liu.se}

\begin{abstract}
Diffusion-attenuated MR signal for heterogeneous media has been represented as a sum of signals from anisotropic Gaussian sub-domains. Any effect of macroscopic (global or ensemble) anisotropy in the signal can be removed by averaging the signal values obtained by differently oriented experimental schemes. The resulting average signal is identical to what one would get if the micro-domains are isotropically (e.g., randomly) distributed, which is the case for ``powdered'' specimens. 
We provide  exact expressions for the orientationally-averaged signal obtained via general gradient waveforms when the microdomains are characterized by a general diffusion tensor possibly featuring three distinct eigenvalues. Our results are expected to be useful in not only multidimensional diffusion MR but also solid-state NMR spectroscopy due to the mathematical similarities in the two fields.
\end{abstract}
\begin{document}

\flushbottom
\maketitle

\thispagestyle{empty}

In MR examinations of porous media as well as biological tissues, one is often confronted with a medium comprising an isotropic distribution of individually anisotropic domains. The effect of diffusion within such media on the MR signal has thus been considered since the 70s \cite{Callaghan79,Mitra92,Joabsson97}. The problem we tackle here is for the situation where diffusion within each microdomain can be taken to be free, thus can be characterized by a microscopic diffusion tensor $\A$. This assumption has been widely employed in the recent development of multidimensional diffusion MR (see Ref.\ \citeonline{Topgaard17perspective} for a recent review and the references therein), which employs general gradient waveforms for diffusion sensitization. The level of diffusion-sensitivity is fully captured by a measurement tensor $\B$ \cite{Westin14miccai}, yielding the signal attenuation 
\begin{align}
S = e^{-\Tr(\A\B)} 
\end{align}
for the microdomain. When some residual anisotropy is present upon the inherent signal averaging over the sample (or voxel in image acquisitions), a series of diffusion signals can be acquired with rotated versions of the same gradient waveforms. Upon averaging such signals, one obtains the orientationally-averaged signal, which is devoid of any macroscopic (ensemble) anisotropy, i.e., the anisotropy of the orientation distribution function of the microdomains. 

As demonstrated in Ref.\ \citeonline{Topgaard17perspective}, there is a close resemblance between the mathematics involved in this problem with that in multidimensional solid-state NMR spectroscopy \cite{SchmidtRohrSpiess_book}. In the latter, the local structure is described by the chemical shift tensor, and a measurement tensor can be introduced, which is determined by the orientation of the main magnetic field and manipulations of the sample orientation within it \cite{Andrew59,Szeverenyi85,Bax83,Ziegler88,Frydman92}.

Our interest in this article is the average signal obtained by repeating a given measurement protocol in all orientations, which is relevant for both diffusion and solid-state MR applications. Specifically, we extend the existing literature \cite{Joabsson97,Eriksson15,Bloembergen53} by providing explicit expressions for the orientational averages to accommodate measurement and/or structure tensors that are $not$ axisymmetric.  

With $\R$ denoting an arbitrary rotation matrix ($\R \R^{\mT} = \R^{\mT} \R = \I$), the complete set of such measurements is spanned by the expression $\R \B \R^\mT$, hence yielding the orientationally-averaged signal as

\begin{subequations}\label{eq:S_avg}
\begin{align}
  \bar{S} = \left \langle e^{-\Tr (\A \R \B \R^\mT)} \right \rangle_{\R} \ , \label{eq:RBRA}
\end{align}
where the average is over the three dimensional special orthogonal
group SO(3), i.e., the space of all rotations. Since the matrix trace
operation is invariant under cyclic permutation of the product in its
argument, this is equivalent to the expression
\begin{align}
  \bar{S} = \left \langle e^{-\Tr (\R^\mT \A \R \B)} \right \rangle_{\R} \ , \label{eq:BRAR}
\end{align}
\end{subequations}
which demonstrates the utility of such an averaging over rotated
protocols: The information is the same as that which would be obtained
by a single measurement, had the specimen consisted of microdomains
with the same diffusivity ($\A$) distributed uniformly in
orientation. For solids, such a specimen is obtained by grinding the
material into a powder, eliminating any nonuniformity in the
orientational dispersion (macroscopic, global, or ensemble anisotropy) the bulk specimen
might have had. Hence we use the terms
\emph{powder}-averaged and \emph{orientationally}-averaged
interchangeably, given the latter effectively achieves the same
result.

An alternative interpretation of the average in \eqref{eq:BRAR} can be realized when one considers the Laplace transform of a function which takes matrix arguments \cite{Herz55}, in our case a tensor distribution $p(\A)$,
\begin{align} \label{eq:Laplace}
{\cal L}_{p(\A)}(\B)= \int_{\mathrm{Sym}_3} e^{-\Tr (\A \B)} \, p(\A) \, \dd\A \ ,
\end{align}
where the integration is performed over $\mathrm{Sym}_3$, the space of all symmetric $3 \times 3$ matrices, and for those matrices $\B \in \mathrm{Sym}_3$ where the integral converges. With $\overline{\cal P}_3$ denoting the set of all positive semi-definite matrices in $\mathrm{Sym}_3$, the applications of interest in this work will be when all the 
matrices  $\A \in \overline{\cal P}_3$. This means that $p(\A)=0$ if $\A$ is not in $\overline{\cal P}_3$, so that the integration in \eqref{eq:Laplace} can be performed over $\overline{\cal P}_3$ rather than $\mathrm{Sym}_3$. ${\cal L}_{p(\A)}(\B)$ will then exist for all $\B \in \overline{\cal P}_3$.

In \eqref{eq:Laplace}, the integration measure is $\dd\A = \prod\dd D_{ij}$ with $1\le i\le j\le 3$. Interestingly, when $p(\A)$ is taken to be the isotropic distribution of a given $\A$ tensor, 
the above expression is equivalent to \eqref{eq:BRAR}. In other words, the orientationally-averaged signal that we are evaluating in this work is nothing but the Laplace transform of the isotropic distribution of a given tensor. Parallels can be drawn with  previous works \cite{Jian07,Scherrer16,Shakya17Dagstuhl} that have employed parametric Wishart distributions (or its generalizations) for representing the detected MR signal; in these studies, the resulting expression for the Laplace transform is borrowed from the mathematics literature \cite{LetacMassam98}. However, the Laplace transform of the isotropic distribution of a general tensor is not available to our knowledge, which is the focus of our work. \CYY{Note that throughout the article, the phrase ``general tensor'' refers to a tensor whose ellipsoidal representation is not axisymmetric. Thus, a general tensor $\mathbf T$ is still to be understood as being real-valued and having the index symmetry $T_{ij}=T_{ji}$.}

Given that the matrices $\A$ and $\B$ are real and symmetric, and that
the averages \eqref{eq:S_avg} are insensitive to individual rotations
of these matrices, we are free to consider their diagonal
forms, possibly in different bases
\begin{align}
\A = \left(\begin{matrix}
a & 0 & 0 \\
0 & b & 0 \\
0 & 0 & c
\end{matrix} \right) \ , \quad
\B = \left(\begin{matrix}
d & 0 & 0 \\
0 & e & 0 \\
0 & 0 & f
\end{matrix} \right) \ . \label{eq:AB}
\end{align}
It turns out that the average signal \eqref{eq:S_avg} can take a
particularly simple form if repeated eigenvalues occur in at least one
of the matrices $\A$ and $\B$. With an ellipsoidal representation of
these matrices in mind, we refer to these cases as \emph{axisymmetric}
or \emph{isotropic}, depending on whether two or all eigenvalues
coincide, respectively. The rank of the matrices appear to have no
further significance for the simplicity of the calculation, except for
rank-1 being a special case of axisymmetry.

In the following sections, we consider various combinations of
symmetries of $\A$ and $\B$ in evaluating the average
\eqref{eq:S_avg}. Keeping in mind that the roles of $\A$ and $\B$ can
be interchanged without changing the average signal, the relevant
cases are
\begin{enumerate}
\item $\A$ general, $\B$ isotropic,
\item $\A$ and $\B$ axisymmetric, $\B$ rank 1,
\item $\A$ and $\B$ axisymmetric,
\item $\A$ general, $\B$ axisymmetric,
\item $\A$ and $\B$ general.
\end{enumerate}


\section*{Results}\label{sec:averages}



Here, we give the resulting expressions for the average signal
\eqref{eq:S_avg} in the cases listed above, with $\A$ and $\B$ given
as in Eq.\ \eqref{eq:AB}. All derivations are provided in \CYY{Supplementary Section I}, 
where $\bar S$ is first derived for the
case (4) of a general $\A$ and axisymmetric $\B$, from which more
specialized cases (1--3) are deduced, while the most general case (5)
is taken up last.

\subsection*{$\A$ general, $\B$ isotropic }\label{sec:AgenBiso}

This is the case where all eigenvalues ($a$, $b$, $c$) of $\A$ are
possibly different, while $\B$ is proportional to identity with $d=e=f$. One
finds,
\begin{equation}\label{eq:isotropic}
\bar S=e^{-d(a+b+c)}=e^{-d\Tr( \A)} \ .
\end{equation}
Indeed, when $\B=d \I$, $\Tr (\A \R \B \R^{\mT}) = \Tr (d \A)$ and
hence the average in Eq.\ \eqref{eq:S_avg} has no effect.

\subsection*{$\A$ and $\B$ axisymmetric, $\B$ rank 1}\label{sec:AaxBaxRank1}

Here, two eigenvalues of $\A$ coincide, which we choose as $b=c$,
while $\B$ is rank-1 with $e=f=0$. This is a widely-utilized case in diffusion MR \cite{Callaghan79,Joabsson97,Yablonskiy02,Kroenke04,Anderson05,Kaden16,Ozarslan18FiP}.
The result follows as,
\begin{equation}\label{eq:AaxBaxRank1}
\begin{split}
\bar S &=\frac{\sqrt{\pi}e^{-c d} }{2} 
\frac{\text{erf} \left(\sqrt{d(a-c)}\right)}{\sqrt{d(a-c)}}
=\frac{\sqrt{\pi}e^{-c d} }{2} 
\frac{\text{erfi}\left(\sqrt{d(c-a)}\right)}{\sqrt{d(c-a)}} \ ,
\end{split}
\end{equation}
where $\text{erf}(\cdot)$ is the error function and
$\text{erfi}(\cdot)$ is the imaginary error function;
$\text{erfi}(x)=\text{erf}(i x)/i$. One can choose either of the formulas
in \eqref{eq:AaxBaxRank1}, but depending on the sign of $d(a-c)$, one
expression will have real arguments, and the other imaginary. For
instance, with $d>0$, $d(a-c)>0$ if $\A$ is prolate and $d(a-c)<0$ if
$\A$ is oblate. It is assumed that $d\neq 0$ and $a\neq c$, as this
implies that one of the matrices is isotropic, but this case can of
course also be included here by continuity arguments.

\subsection*{$\A$ and $\B$ axisymmetric}\label{sec:AaxiBaxi}
When both matrices have axisymmetry, with $b=c$ and $e=f$, we have \cite{Eriksson15}
\begin{equation}\label{eq:AaxBax}
\begin{split}
\bar S=\frac{\sqrt{\pi}e^{-c d-f(a+c)} }{2} 
\frac{\text{erf}\left(\sqrt{(a-c)(d-f)}\right)}{\sqrt{(a-c)(d-f)}}
=\frac{\sqrt{\pi}e^{-c d-f(a+c)} }{2} 
\frac{\text{erfi}\left(\sqrt{(c-a)(d-f)}\right)}{\sqrt{(c-a)(d-f)}}.
\end{split}
\end{equation}
With similar remarks as in \CYY{the previous case}, any of the
two forms in \eqref{eq:AaxBax} can be chosen, and if moreover $a=c$ or $d=f$,
one gets the \CYY{first (isotropic) case above}.  The first
form has real arguments when both $\A$ and $\B$ are either prolate or
oblate, while the second form may be preferable when one matrix is
prolate and the other oblate.

\subsection*{$\A$ general, $\B$ axisymmetric} \label{sec:agenbaxi}
For the almost-general case where the only condition is axisymmetry of
$\B$ with $e=f$, the result is given in four alternative forms by
\begin{subequations} \label{eq:agenbaxi}
\begin{align} 
  \bar S & = 
  \begin{cases}
    e^{-f \Tr(\A)} \frac {\sqrt{\pi} e^{(f-d) c}} {2} \sum\limits_{n=0}^\infty
    \frac {[(c-a)(d-f)]^n } {\left(n+1/2\right)!} \,
    {}_2F_1 \!\left( \tfrac{1}{2}, -n; 1; \tfrac{a-b}{a-c} \right),
    & a\neq c\quad  \\
    \frac {\sqrt{\pi} e^{-ad -(a+b)f}} {2} 
    \frac {\text{erf} \left(\sqrt{(b-a)(d-f)} \right)} {\sqrt{(b-a)(d-f)}}
    =\frac {\sqrt{\pi} e^{-ad -(a+b)f}} {2}
    \frac {\text{erfi} \left(\sqrt{(a-b)(d-f)}\right)} {\sqrt{(a-b)(d-f)}},
    & a=c \mbox{ \CYY{(case 3)} }
  \end{cases}\\
  &= e^{-f \Tr(\A)} e^{(f-d)c}\sum \limits_{n=0}^\infty 
  \frac {\left[(a-b)(d-f)\right]^n} {(2n+1)n!}  \,
  {}_1F_1 \!\left(n+1;n+\tfrac{3}{2};(c-a)(d-f) \right) \\
  &= e^{-f \Tr(\A)} \frac {\sqrt{\pi} e^{(f-d)c}} {2} \sum_{n=0}^\infty
  \frac {[(c-a)(d-f)]^n} {(n+1/2)!} \,
  {}_2F_2 \!\left(\tfrac{1}{2},n+1;1,n+\tfrac{3}{2};(a-b)(d-f) \right) \\
  &= e^{-f \Tr(\A)} \frac {\sqrt{\pi} e^{(f-d)c}} {2} \sum _{n=0}^\infty
  \frac {[(a-b)(d-f)]^{2n}} {16^{n} \left(2n+\frac{1}{2}\right)!}  \binom{2 n}{n} \,
  {}_1F_1 \!\left(2n+1;2n+\tfrac{3}{2};\tfrac{1}{2} (f-d) (a+b-2 c)\right) .
\end{align}
\end{subequations}
Here, the symbol ${}_iF_j(\cdot)$ stands for a hypergeometric function, in particular with ${}_1F_1 (\cdot)$ being
the confluent hypergeometric function. Eqs.\ (\ref{eq:agenbaxi}a--c) were derived by a rather direct approach; see \CYY{Supplementary Section I.A.} 
The fourth expression (\ref{eq:agenbaxi}d), which is perhaps the most interesting  for $\bar S$ when one matrix is axisymmetric, is derived from the general expression, i.e, the case accounted for \CYY{next}. As discussed \CYY{below}, it can lead to very efficient numerical evaluation.  

The expressions above feature hypergeometric functions. However, it is possible to express them using more familiar functions. For example, by using a scheme \cite{Ozarslan_NI06} that involves term by term comparison with the confluent hypergeometric function's asymptotic behaviour \cite{Abramowitzbook}, we were able to write Eq.\ \eqref{eq:agenbaxi}b in the following form:
\begin{align}
\bar S & = e^{-f \Tr(\A)} \sum_{k=0}^\infty \tfrac{(2k)!}{(k!)^2 4^k} \left[ {(a-b)(d-f)} \right]^k  
\left\{ e^{(f-d)c} \sum_{n=1}^k \tfrac 1 {[2(a-c)(d-f)]^{n}} \sum_{j=0}^{k-n}
\tfrac {(-1)^{j+1} (2n+2j-1)!!} {(k-n-j)! (n+j)! (2j+1)!!} \right. \nonumber \\
& \hspace{16em} \left. + e^{(f-d)a}
\tfrac {\sqrt{\pi} \, \mathrm{erfi} \left( \sqrt{(a-c)(d-f)} \right)}
{2 \sqrt{(a-c)(d-f)}} \sum_{n=0}^k \tfrac{(2n-1)!!}{n! (k-n)!}
\tfrac 1 {[ 2(a-c)(d-f)]^{n}} \right\} \ ,
\end{align}
where $(2n+1)!!$ denotes the product of all odd
numbers from 1 up to $(2n+1)$. Similarly, Eq.\ (\ref{eq:agenbaxi}a)
can be expressed via more familiar functions, using
Eq.\ \eqref{eq:legendreform}.

\CYY{We also note that in the literature the signal for this case is reduced to a one-dimensional integral wherein the integrand is given by a complete elliptic integral \cite{Bloembergen53,SchmidtRohrSpiess_book}.}

\subsection*{$\A$ and $\B$ general}\label{sec:agenbgen}
In the general case where neither matrix (necessarily) has any degenerate eigenvalue, the result can be expressed in one of the three alternative forms
\begin{subequations} \label{eq:agenbgen}
\begin{align} 
&\bar S=e^{-\frac{1}{2}(a+b)(e+f)-cd} \hspace*{-4mm}\sum_{0\leq m \leq k < \infty}\hspace*{-4mm}q_{mk}Y_{mk} \ ,
\end{align} 
where 
\begin{align} 
\begin{split}
&q_{mk}=
\begin{cases}
\begin{cases}
\frac{((c-a)(e-f))^k}{4^k ((k/2)!)^2}, &m=0 \mbox{ and $k$ even} \\
0, & \mbox{ otherwise}
\end{cases},& a=b\\
\frac{1}{2^k}(c-a)^{k-m}\sum\limits_{j=\left \lfloor{\frac{m+1}{2}}\right \rfloor}^{k/2}
I_{k\! - \! j}\left(\! \frac{(e-f)(a-b)}{2}\right)(a-b)^{m-j}(e-f)^{k-j}\frac{1}{j!(k-2j)!}\binom{2j}{m}
,& a+b-2c=0\\
\frac{\left(c-a\right)^{k-m}}{2^k (k-m)!}
\sum\limits_{j=0}^{k/2}I_{k\! - \! j}\left(\! \frac{(e-f)(a-b)}{2}\right)\frac{(a\! +\! b\! - \! 2 c)^{m-2j}
(a-b)^j(e-f)^{k-j}}{j!}
\! _2\widetilde F_1(m\! - \!k,-2j;1\!+\! m-2j,\frac{a\! +\! b\! -\! 2 c}{a-b}), &\mbox{otherwise}
\end{cases}
\end{split}
\end{align}
and there are three alternatives for $Y_{mk}$ as follows:
\begin{align} 
&Y_{mk}^{(1)}=
\begin{cases}
\frac{\sqrt{\pi}4^{-m}}{2}\binom{2 m}{m}\sum\limits_{n=0}^\infty\frac{(\frac{1}{2}(c-a)(2d-e-f))^n  (k+n)! \,
   _2F_1\left(m+\frac{1}{2},-n;m+1;\frac{a-b}{a-c}\right)}{n! (k+n+1/2)!}, & a\neq c\\
  \frac{\sqrt{\pi}4^{-m}}{2}\binom{2 m}{m}\ \frac{k! }{\left(
   k+1/2\right)!}\, _3F_2\left(1,k+1,m+\frac{1}{2};k+\frac{3}{2},m+1; \frac{1}{2}(a-b)(2d-e-f)\right), & a=c
\end{cases} \\
&Y_{mk}^{(2)}=\frac{\sqrt{\pi}4^{-m}}{2}\sum_{n=0}^\infty\tfrac{(\frac{1}{2}(a-b)(2d-e-f))^n (k+n)! \binom{2 (n+m)}{n+m} \,
   _1F_1\left(k+n+1;k+n+\frac{3}{2};\frac{1}{2}(c-a)(2d-e-f)\right)}{4^n n! \left(k+ n+1/2\right)!},\\
&Y_{mk}^{(3)}=\frac{\sqrt{\pi}4^{-m}}{2}\binom{2 m}{m}
\sum_{n=0}^\infty\tfrac{(\frac{1}{2}(c-a)(2d-e-f))^{n } (k+n )! \,
   _2F_2\left(m+\frac{1}{2},k+n +1;m+1,k+n +\frac{3}{2};\frac{1}{2}(a-b)(2d-e-f)\right)}{n !
   \left(k+n +1/2\right)!} \ .
\end{align}
\end{subequations}
The coefficients $q_{mk}$ contain $I_m (\cdot)$, i.e., the modified
Bessel function of the first kind for various orders $m$, and the
regularized hypergeometric function ${}_2\widetilde F_1 (\cdot)$. Also
note that the floor function $\lfloor \cdot \rfloor$ indicates the
largest integer smaller than or equal to its argument.

It is straightforward to check that $\bar S$ as given by
Eq.\ \eqref{eq:agenbgen} is invariant under the rescaling $\A \to
\lambda \A, \B \to \tfrac{1}{\lambda}\B$ for nonzero $\lambda$,
verifying an obvious invariance it should have due to its definition
\eqref{eq:S_avg}. Implied by that definition, $\bar S$ must also be
unaffected by permutations of $\{a,b,c\}$ and of $\{d,e,f\}$, as well
as swapping $\{a,b,c\} \leftrightarrow \{d,e,f\}$. This is not at all
evident in the expansions above; in fact the ordering of eigenvalues
can have drastic effects on the numerical behaviour of the series,
even though the eventual sum does not change. Given $\A$ and $\B$, and
their eigenvalues, the remaining task is thus to assign them (in some order) to $a,b,
c$ and $d,e,f$ in such a way that the series can be approximated well
with only a few terms; see the next section.


\section*{Numerical Behaviour}\label{sec:implementation}

In this section we give some comments and guidelines on how to
evaluate $\bar S$ via the series in Eqs.\ \eqref{eq:agenbaxi} and
\eqref{eq:agenbgen}. The numerical behaviour of the series given in
Eqs.\ \eqref{eq:agenbaxi} and \eqref{eq:agenbgen} varies drastically
depending on how the eigenvalues of $\A$ and $\B$ are ordered (i.e.,
which eigenvalue is named $a$, $b$, and so on). In the formulas in
Eqs.\ \eqref{eq:agenbgen}, there is also the option of interchanging
$\A$ and $\B$. \CYY{The discussion here revolves around an example
employing Eq.\ \eqref{eq:agenbaxi}, but the guidelines apply to the
general case of Eq.\ \eqref{eq:agenbgen} as well, albeit more
involved.}



When $\B$ is axisymmetric ($e=f$), $f$ and $d$ are fixed in Eqs.\ 
\eqref{eq:agenbaxi}. On the other hand, in \CYY{these formulas}, 
which all give the same result, one is free to permute $\{a,b,c\}$. Although this does not
affect the answer, it affects the number of terms needed to get a good
approximation. Hence, given the three eigenvalues of $\A$, one wants
to assign them to $a,b,c$ in a clever way, as well as choose the most
efficient formula.

Looking at the expansions \eqref{eq:agenbaxi}, while there appears no
obvious ``best'' choice for assigning a given set of eigenvalues, it is
still wise to try and (\emph{i}) keep the series expansion parameter (e.g., $(a-b)(d-f)$ in Eq. \ref{eq:agenbaxi}d) small, 
and (\emph{ii}) avoid alternating signs, since an expansion where all
terms have the same sign cannot suffer from cancellation effects. For
instance, it is worth paying attention to the magnitude and sign of
$(c-a)(d-f)$ and ${}_2 F_1 \!\left( \tfrac 12,-n;1; \tfrac {a-b}{a-c}
\right)$ in Eq.\ (\ref{eq:agenbaxi}a), and so on for the
rest. Therefore, even though a deep analysis of the properties of
hypergeometric functions is beyond the scope of this work, it is
useful to note the following as a guideline for their sign and
magnitude. With $n \in \mathbf{N}$, the functions ${}_2F_1 \!\left(
\frac 12, -n; 1; x \right) > 0$ are positive and decreasing functions
of $x<0$, and happen to be polynomials. On the other hand, $_2F_2
\!\left(\frac{1}{2},n+1;1,n+\frac{3}{2};x\right) >0$ are positive and
increasing functions of $x \geq 0$. Finally, the functions $_1F_1
\!\left(n+1;n+\frac{3}{2};x \right)>0$ are positive and increasing for
all $x \in \mathbf{R}$. (See \CYY{Supplementary Section II} 
for some elaboration on this.)

In Figures \ref{fig:f21}--\ref{fig:S4}, the contributions of
individual terms in the series are displayed, with the index of the
terms running on the horizontal axis. \CYY{The figures correspond, respectively, to the alternative forms (\ref{eq:agenbaxi}a) and (\ref{eq:agenbaxi}d), which were chosen as examples to illustrate clearly} that term by term each series exhibits different
behaviours depending on the allocation of the eigenvalues to the
parameters $a,b,c$. For the particular case presented in the figures,
the matrix $\A$ has eigenvalues $0.1,0.2,3$ and the matrix $\B$ has
eigenvalues $6, 0.5, 0.5$, for which $\bar S \approx .019175$, and all
the terms are normalized by $\bar S$ so that their sum is $1$.

\begin{figure}[htbp]
\begin{center}
\includegraphics[scale=.33]{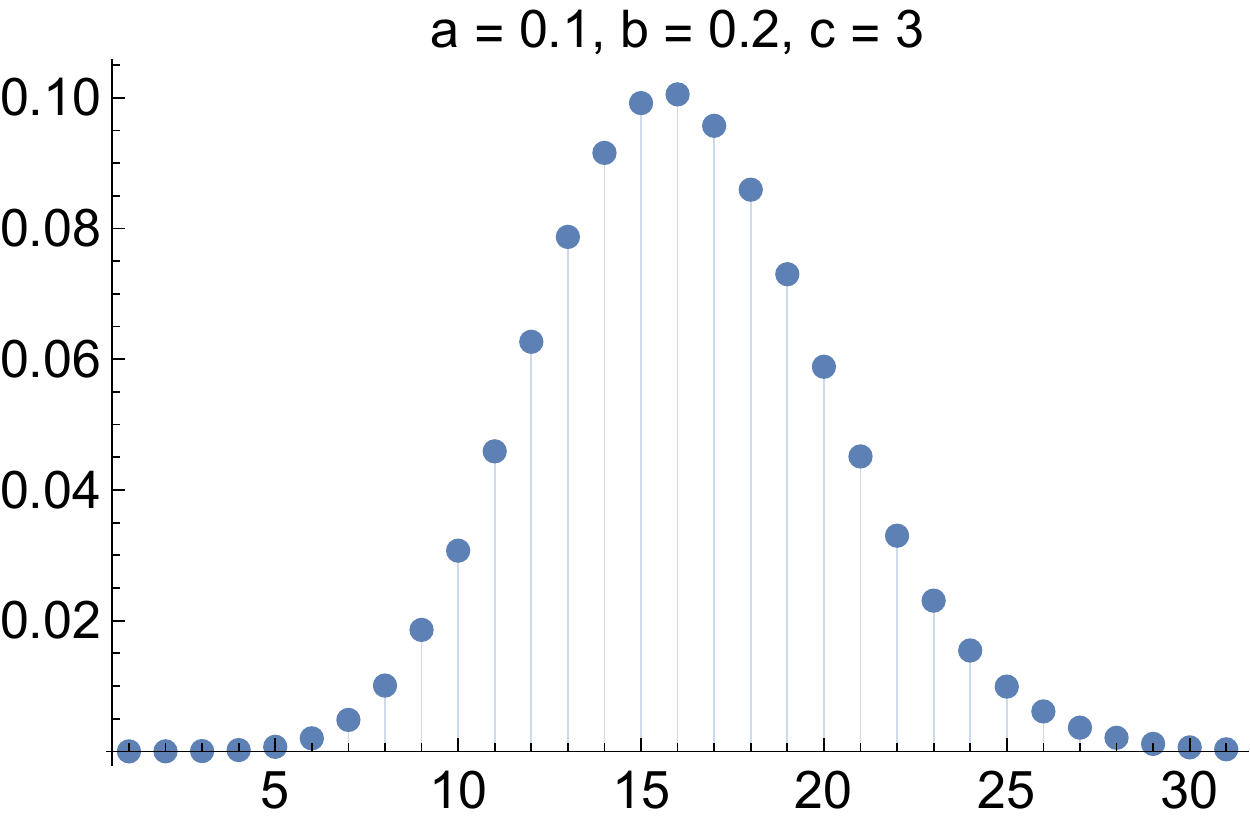}
\includegraphics[scale=.33]{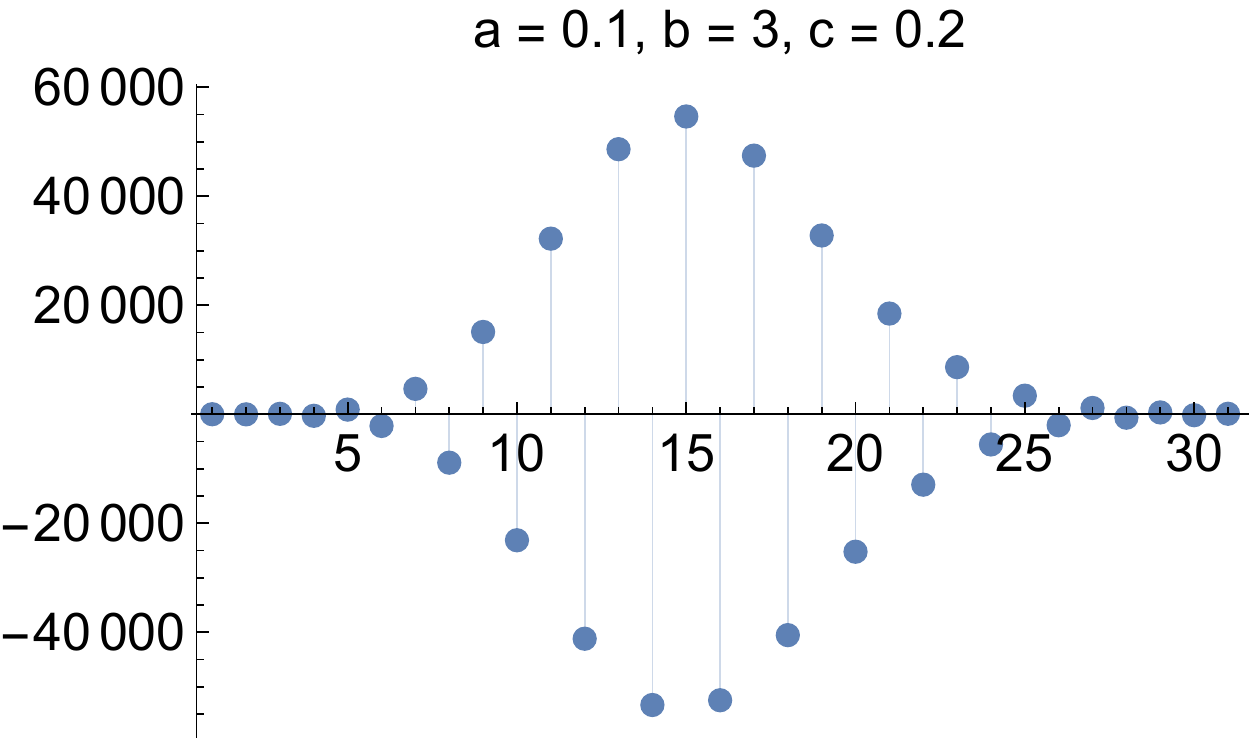}
\includegraphics[scale=.33]{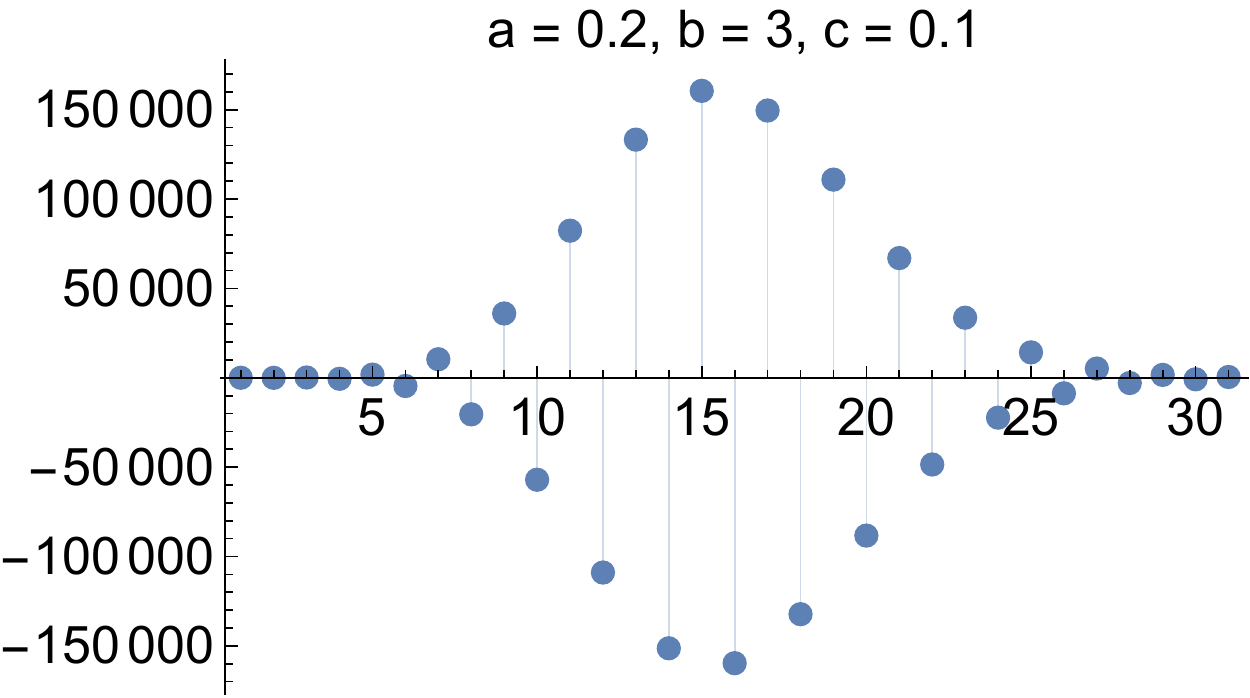}
\caption{The first terms in the series in Eq.\ (\ref{eq:agenbaxi}a), which uses the hypergeometric function ${}_2F_1$. Terms are normalized by $\bar S$ so that they add up to $1$. $\B$ has eigenvalues $d=6$, $e=f=0.5$, while $\A$ has eigenvalues $0.1,0.2,3$. Depending on how the latter are assigned to $a,b$ and $c$, the trend of the series expansion varies. \CYY{Displayed are three (instead of six) distinct choices, because} each term is invariant under the change $a \leftrightarrow b$, which is not obvious from Eq.\ (\ref{eq:agenbaxi}a), but can be shown using Eq.\ \eqref{eq:legendreform}.}
\label{fig:f21}
\end{center}
\end{figure}

The first alternative (\ref{eq:agenbaxi}a) has the property that its
individual terms are invariant under the change $a \leftrightarrow
b$, \CYY{which can be shown to follow from relation
  \eqref{eq:legendreform}.} Therefore, out of the six possible
assignments between the elements of the sets $\{a,b,c\}$ and
$\{0.1,0.2,3\}$ only three are distinct. These are what is depicted in
Figure \ref{fig:f21}. It is seen clearly that for these parameters
this alternative does not afford a series expansion that can be
truncated after the first few terms for a usable result. The terms
making the most significant contribution to the sum are not even at
the beginning, but in this example rather around term number 15. Furthermore, the two
latter choices exhibit terms of alternating sign and magnitudes of
about $10^5$ times the sum itself.

\begin{figure}[htbp]
\begin{center}
\includegraphics[scale=.33]{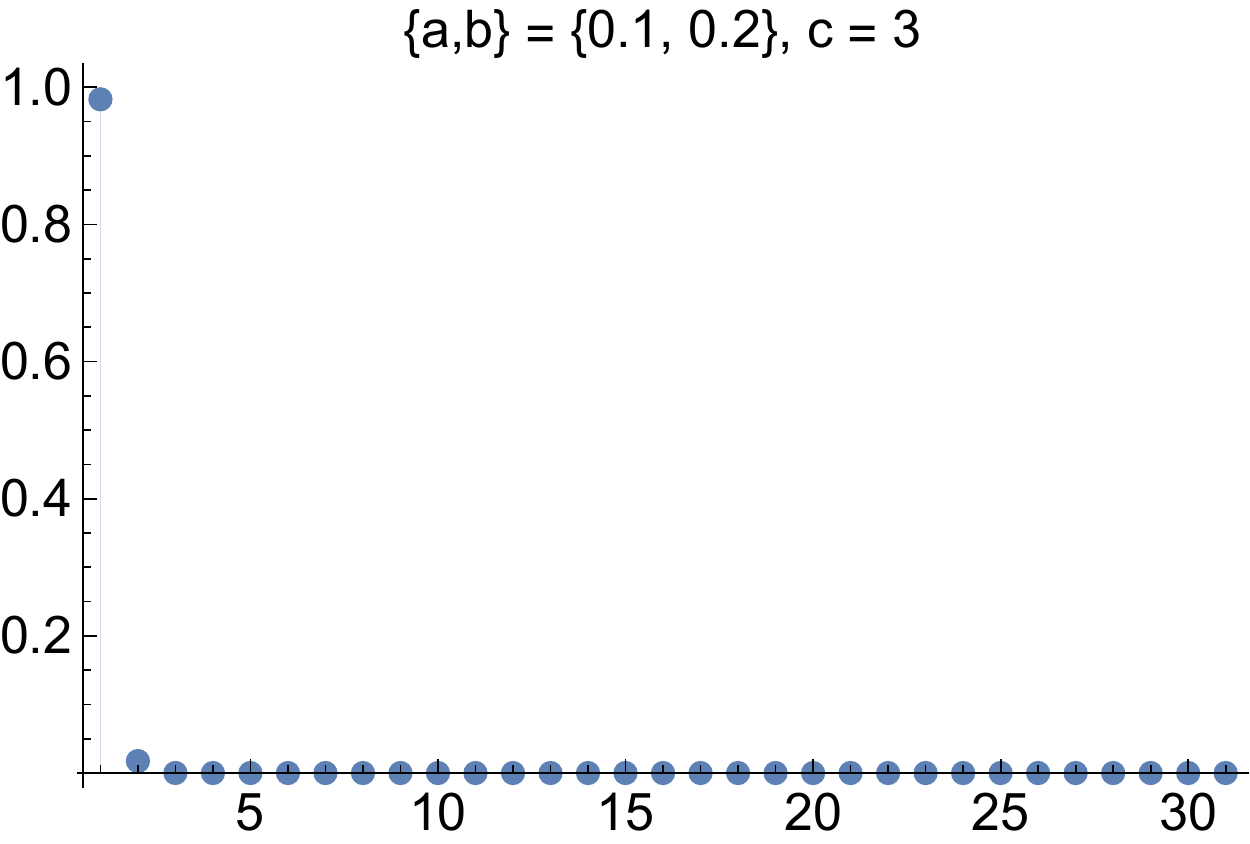}
\includegraphics[scale=.33]{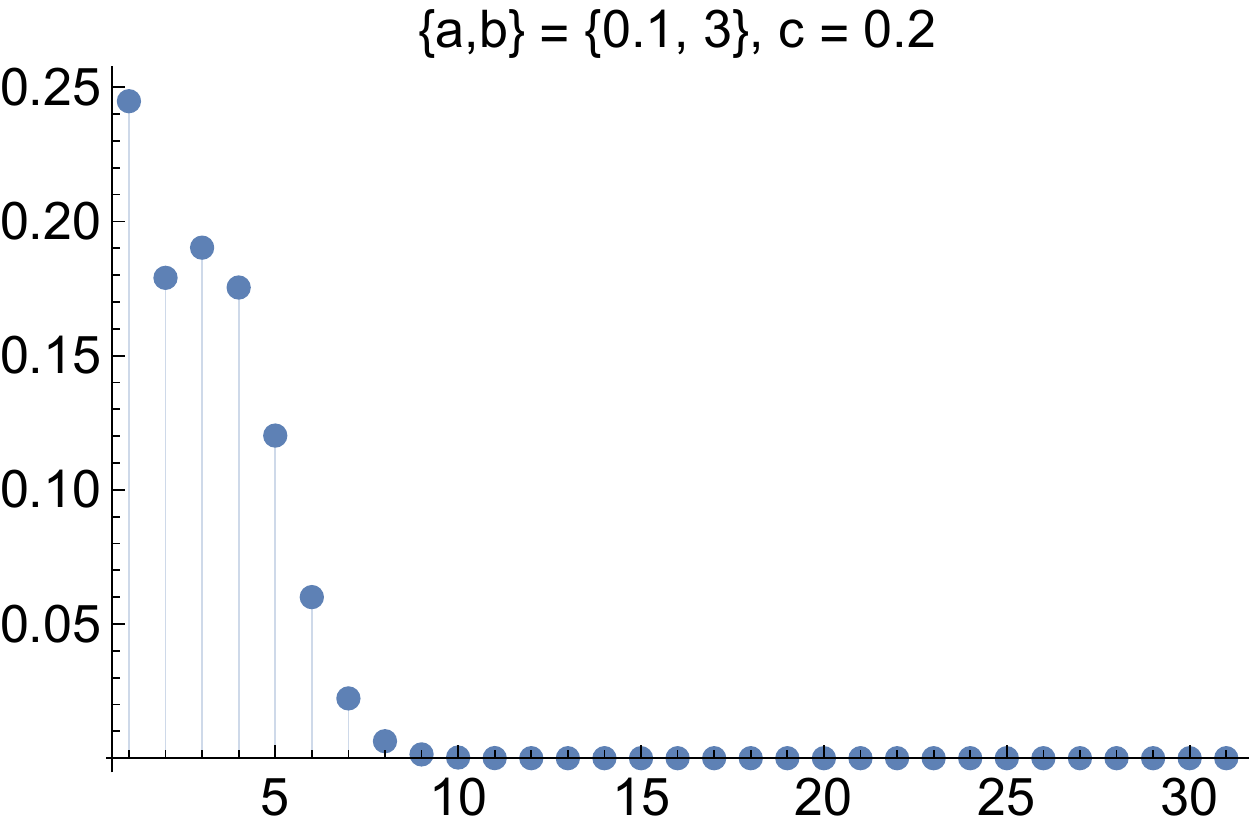}
\includegraphics[scale=.33]{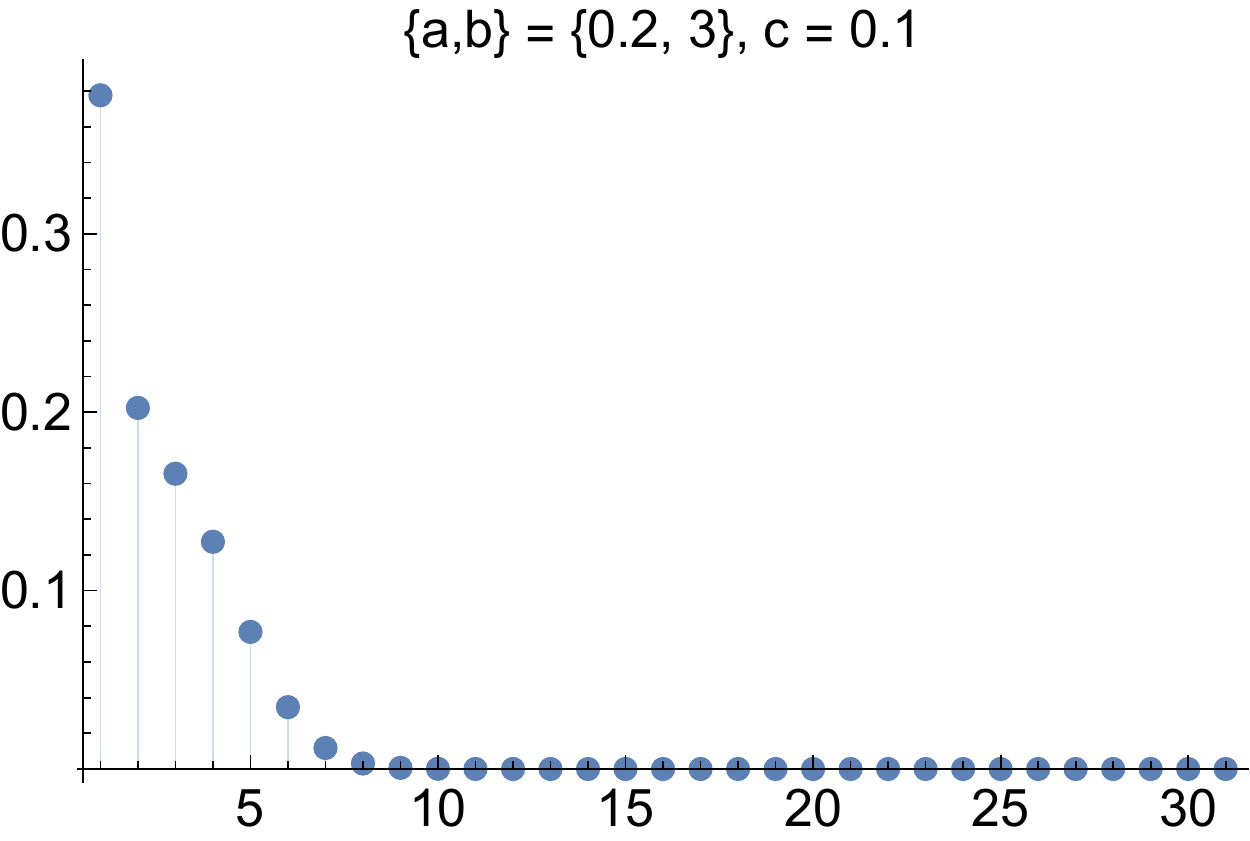}
\caption{The first terms in the series in Eq.\ (\ref{eq:agenbaxi}d),
  which uses the hypergeometric function ${}_1F_1$. Terms are
  normalized by $\bar S$ so that they add up to $1$. $\B$ has
  eigenvalues $d=6$, $e=f=0.5$, while $\A$ has eigenvalues
  $0.1,0.2,3$. Depending on how the latter are assigned to $a,b$ and
  $c$, the trend of the series expansion varies. \CYY{Displayed are
    three (instead of six) distinct choices, since} individual terms
  of Eq.\ (\ref{eq:agenbaxi}d) are manifestly invariant under swapping
  $a \leftrightarrow b$.}
\label{fig:S4}
\end{center}
\end{figure}

\CYY{Figure \ref{fig:S4} depicts all three distinct choices for the
  alternative expression (\ref{eq:agenbaxi}d). We see a very desirable
  feature here. When the largest value is assigned to $c$, the terms
  are seen to converge quickly (note that the expression is invariant
  under the exchange $a \leftrightarrow b$). With this choice (e.g.,
  $a=0.2, b=0.1, c=3$) we have ${}_1
  F_1(2n+1;2n+3/2;(f-d)(a+b-2c)/2)={}_1 F_1(2n+1;2n+3/2;15.675)$ which
  is bounded with respect to $n$. In fact, it follows from termwise
  comparison in the defining series for $_1F_1$ (see Supplementary
  Section II) 
  that $_1F_1 (2n+1, 2n+3/2;
  15.675) \leq {}_1F_1 (2n+3/2;2n+3/2;15.675)=e^{15.675}$. We also see
  that only even powers $(a-b)(d-f)=0.55$ enter, resulting in all
  terms of the series being positive; no cancellation occurs. This
  choice leads to a sum that converges so quickly that $\bar S$ is
  approximated within $2\%$ by only the first term, while the first
  two terms attain an error less than 0.01\%.  In addition, it should
  also be noted that for moderate values of $n$, ${}_1
  F_1(n+1;n+3/2;x)$ are explicitly expressible in elementary functions
  and that for ${}_1F_1$ there are effective recursive relations at
  hand (see Supplementary Section II).} 

\paragraph{Swapping $a$ and $b$}
By construction, all series are unaffected by permutations of $a,b,c$
(although it affects the numerical behaviour) but in the series
(\ref{eq:agenbaxi}a), also the individual terms are unaffected by the
change $a \leftrightarrow b$.  This follows from the relation
\begin{equation}\label{eq:legendreform}
{}_2F_1 \!\left( \tfrac{1}{2}, -n; 1;z\right)=(1-z)^{n/2} P_n\!\left(\frac{2-z}{2 \sqrt{1-z}}\right) \ , 
\end{equation}
where $P_n(.)$ is the $n$th order Legendre polynomial. It is then readily verified that $(a-c)^n\, {}_2F_1\left(
\tfrac{1}{2}, -n; 1; \tfrac{a-b}{a-c}\right)$ and $(b-c)^n\,
      {}_2F_1\left( \tfrac{1}{2}, -n; 1; \tfrac{b-a}{b-c}\right)$ are
      indeed equal.
The terms in the series (\ref{eq:agenbaxi}d) are also unaffected
by the change $a \leftrightarrow b$, but this is immediate from
the expression.

\section*{Applications}

We have presented explicit formulas for the orientationally-averaged signal $\bar S$ in Eq.\  \eqref{eq:S_avg} for general (symmetric) matrices $\A$ and $\B$. These formulas complement the well-known cases 1, 2, and 3, i.e., the formulas in Eqs.\ \eqref{eq:isotropic}--\eqref{eq:AaxBax}, where $\A$ and $\B$ have various symmetries. However, even when one matrix, say $\B$, is rank-1, the signal expression for a general $\A$ (Eq.\ \eqref{eq:agenbaxi}) is believed to be new. 
Eq.\ \eqref{eq:agenbgen} raises the question of the usefulness of case 5, in which both measurement and diffusion tensors are general.
\EOO{We argue that this solution is indeed useful, for example, when a powdered specimen characterized by a general diffusion tensor is imaged. In MRI, the imaging gradients lead to a rank-3 measurement tensor even when a standard Stejskal-Tanner sequence is employed. Moreover, } 
since for a general $\A$, even using a rank-1 measurement tensor $\B=\begin{psmallmatrix} d & 0 & 0 \\0 & 0 & 0 \\ 0 & 0 & 0 \end{psmallmatrix}$, the knowledge of $\bar S(d)$ for sufficiently many $d$ determines $\A$. \CYY{Below,} we also discuss power-laws and asymptotic behaviour of the signal in a more general setting. 

From a measurement side, many ``independent'' measurements are necessary in order to mitigate noise-related effects.
With an axisymmetric measurement tensor $\B$, the space of measurements is two dimensional (two degrees of freedom in $\B$), while in the general case the 
space of measurements is three dimensional. \emph{Loosely speaking}, this means that in the latter case, there are more ``independent'' measurement tensors available
close to the origin of this space (or any other point common to both spaces, for that matter). Since measurements close to the origin (i.e., $\B$ tensors with small eigenvalues) produce higher signal value, this is favorable from a signal-to-noise ratio perspective.
On the theoretical side, however, there are situations where a general $\B$ tensor actually is crucial in determining the 
diffusivity properties of the specimen; see \CYY{subsection below on the estimation of $\A$ from the series expansion of $\bar{S}$.} 

\subsection*{Remarks on the white-matter signal at large diffusion-weighting}\label{sec:signalwhitematter}

Recently, the orientationally-averaged signal in white-matter regions of the brain was observed to follow the power-law $\bar S \propto d^{-1/2}$ at large $d$ when $e=f=0$, e.g., in traditional Stejskal-Tanner measurements \cite{McKinnon17,Novikov18modeling}. Because such decay is predicted for vanishing transverse diffusivity, these results have been interpreted to justify the ``stick'' model of axons \cite{Behrens03,Kroenke04,Zhang12} ($a>b=c=0$) while also suggesting that the signal from the extra-axonal space disappears at large diffusion-weighting. In a recent article \cite{Ozarslan18FiP}, we showed that such a decay is indeed expected for one-dimensional curvilinear diffusion observed via narrow pulses. On the other hand, in acquisitions involving wide pulses, the curvature of the fibers has to be limited in order for $\bar S \propto d^{-1/2}$ behaviour to emerge. We also noted that the $\bar S \propto d^{-1/2}$ dependence is valid for an intermediate range of diffusion weightings as the true asymptotics of the signal decay is governed by a steeper decay \cite{Ozarslan18FiP}. 

Here, based on these findings, we consider a rank-1 diffusion tensor for representing  intra-axonal diffusion, which is capable of reproducing the intermediate $\bar S (d,0,0) \propto d^{-1/2}$ dependence for measurement tensors of the form  $ \B=\begin{psmallmatrix} d & 0 & 0 \\0 & 0 & 0 \\ 0 & 0 & 0 \end{psmallmatrix}$. For rank-2, axisymmetric measurement tensors, $ \B=\begin{psmallmatrix} d & 0 & 0 \\0 & d & 0 \\ 0 & 0 & 0 \end{psmallmatrix}$, the orientationally-averaged decay is characterized by the power-law $\bar S_a(d,d,0) \propto d^{-1}$ with $a$ being the single nonzero eigenvalue of $\A$\CYY{, which can be shown to follow from Eq.\ \eqref{eq:AaxBaxRank1} upon interchanging $\A$ and $\B$}. 
For non-axisymmetric, rank-2 measurement tensors, the signal, $\bar S (d,e,0)$ with $d>e$ obeys the relationship
\begin{align}
\bar{S}_a(d,d,0) \le \bar{S}_a(d,e,0) \le \bar{S}_a(e,e,0) \ .
\end{align}
Now, consider the case in which the measurement is tuned by ``inflating'' the $\B$-matrix while preserving its shape, i.e., varying $d$ while keeping $e/d$ fixed. Then, both the lower and upper bounds indicated in the above expression decay according to $\bar S_a \propto d^{-1}$. To satisfy the inequality, it must thus also be that $\bar{S}_a(d,e,0) \propto d^{-1}$ , which establishes the decay of the orientationally-averaged signal decay obtained via \emph{general} rank-2 ($\B$ matrices). Thus, an alternative validation of the stick model could be performed by acquiring data using planar encoding, in which case the expected decay of the orientationally averaged signal would be characterized by a decay $\propto d^{-1}$ regardless of whether or not the $\B$ matrix is axisymmetric.

\subsection*{Observation of the (non)axisymmetry of local diffusion tensors}\label{sec:nonaxisymmetry}

Here, we focus our attention to rank-1 measurement tensors, $ \B=d\begin{psmallmatrix} 1 &0 & 0 \\0 & 0 & 0 \\ 0 & 0 & 0 \end{psmallmatrix}$, e.g., obtained using the Stejskal-Tanner sequence, and investigate to what extent  $\A$ is determined when $\bar S=\bar S_{\A}(d)$ is known for some values $d$. Specifically, we ask the question:  Is it possible to distinguish an axisymmetric $\A$ from a non-axisymmetric $\A$? 

\begin{figure}[h!]
\begin{center}
\includegraphics[trim={1cm 0 0 0},clip,scale=.5]{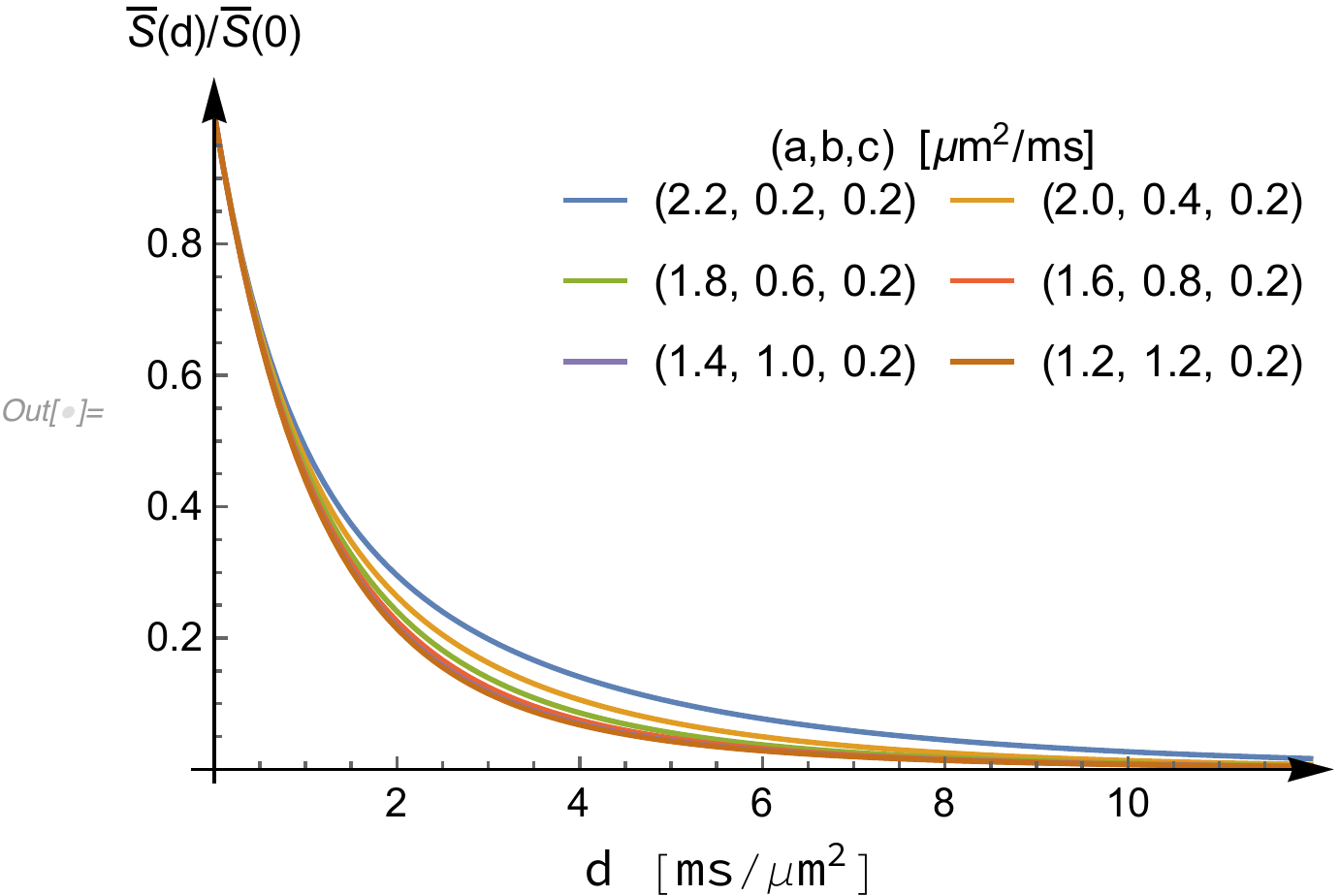}
\caption{Given the rank-1 measurement tensor $\B$, with eigenvalues $(d,0,0)$, the signal $\bar S_{\A_i}(d)$ is plotted as a function of $d$ for six different diffusion tensors $\A_1, \A_2, \ldots \A_6$. These diffusion tensors all have the same trace, and their eigenvalues $(a,b,c)$ are in order
$\A_1:(2.2, 0.2, 0.2)$ $\mu$m${}^2/$ms,
$\A_2:(2.0, 0.4, 0.2)$ $\mu$m${}^2/$ms,
$\A_3:(1.8, 0.6, 0.2)$ $\mu$m${}^2/$ms,
$\A_4:(1.6, 0.8, 0.2)$ $\mu$m${}^2/$ms,
$\A_5:(1.4, 1.0, 0.2)$ $\mu$m${}^2/$ms,
$\A_6:(1.2, 1.2, 0.2)$ $\mu$m${}^2/$ms.
} \label{fig:decay1}
\end{center}
\end{figure}
The answer to this question is yes; in fact, it is not so hard to see that if $\bar S_{\A_1}(d)$ and $\bar S_{\A_2}(d)$ are equal for all $d \geq 0$, then $\A_1$ and $\A_2$ must be equal---see for instance the remark near Eq.\ \eqref{eq:Brank1} in the next section. \CYY{(Also recall that we identify the matrices according to their eigenvalues; two matrices are ``equal'' for the purposes of this article if one can be rotated into the other.)} 
As an illustration, starting with a rank-1 measurement tensor $ \B=d\begin{psmallmatrix} 1 &0 & 0 \\0 & 0 & 0 \\ 0 & 0 & 0 \end{psmallmatrix}$, Figure \ref{fig:decay1} shows the signal $\bar S$ as a function of $d$ for six different diffusion tensors $\A_1, \A_2, \ldots \A_6$. In this example, all tensors $\A_i$, $i=1,2,\ldots 6$ have the same trace yielding the same behaviour near the origin, and their common smallest eigenvalue leads to the same large $b$ behaviour. By varying the remaining two eigenvalues, (whose values are found in Figure \ref{fig:decay1}) while keeping their sum constant, different curves $\bar S_{\A_i}(d)$ are produced. The six curves shown in Figure \ref{fig:decay1} are samples from a family (indexed by the diffusion tensor $\A$) of curves characterized by their initial and asymptotic behaviour. It is interesting to note that no axisymmetric diffusion tensor $\A$ other
than $\A_1$ and $\A_6$ can produce a curve in this family. In this context, the `same asymptotic behaviour' refers to the same exponential decay as $d\to \infty$, and this is given by the smallest eigenvalue of the diffusion tensor.
The initial behaviour, i.e., $\bar S'(0)$ is given by the trace of the diffusivity, and 
it is immediate to see that an axisymmetric tensor where the smallest eigenvalue is $0.2$\nolinebreak\ $\mu$m${}^2/$ms and where the trace is $2.6$ $\mu$m${}^2/$ms
has to have the eigenvalues of either $\A_1$ or $\A_6$.

In the remaining part of this section, we consider the intermediate $d$ regime alluded to in the previous section. For a general rank-3 diffusion tensor $\A=\begin{psmallmatrix} a & 0 & 0 \\ 0 & b & 0 \\ 0 &
0 & c \end{psmallmatrix}$, the falloff is exponential, which can be inferred from the expression (for $a\ge b \ge c>0$)
\begin{align}
\bar S_{a,a,c}(d)
\leq
\bar S_{a,b,c}(d)
\leq 
\bar S_{a,c,c}(d)
\end{align}
since both the left- and the right-hand-sides of the above expression decay exponentially fast; see \eqref{eq:AaxBaxRank1}. Consequently, a reliable inference cannot be made in the large $d$ regime in typical acquisitions due to limited SNR. The only exception is when $c=0$, i.e., when $\A$ is of rank 2. In this case, the arguments from the preceding section can be employed by interchanging the matrices $\A$ and $\B$, i.e., 
\begin{align}
\bar S_{a,a,0}(d)
\leq
\bar S_{a,b,0}(d)
\leq 
\bar S_{b,b,0}(d) \ ,
\end{align}
where, $\bar S_{a,a,0}(d) = (2ad)^{-1}$ for large $d$. 
Hence we expect \CYY{(at least when $a>b$)}
\begin{align}
\frac{1}{2ad}
\leq
\bar S_{a,b,0}(d)
\leq 
\frac{1}{2bd} \ \mbox{for large } d.
\end{align}
This is indeed the case, as the following theorem shows.
\begin{theorem} \label{thm:kappa}
Suppose that the diffusion tensor $\A$ has eigenvalues $a,b,0$ where $a,b>0$, and that the measurement tensor $\B$ has eigenvalues $d,0,0$.
Denoting the corresponding powder average with $ \bar S_{a,b,0}(d)$, it then holds that
$$
\lim_{d \to \infty}  d\bar S_{a,b,0}(d)=\frac{1}{2\sqrt{a b}}.
$$
\end{theorem}
\begin{proof}
See \CYY{Supplementary Section V}. 
\end{proof}

\begin{figure}[h!]
\begin{center}
\includegraphics[trim={1cm 0 0 0},clip,scale=.5]{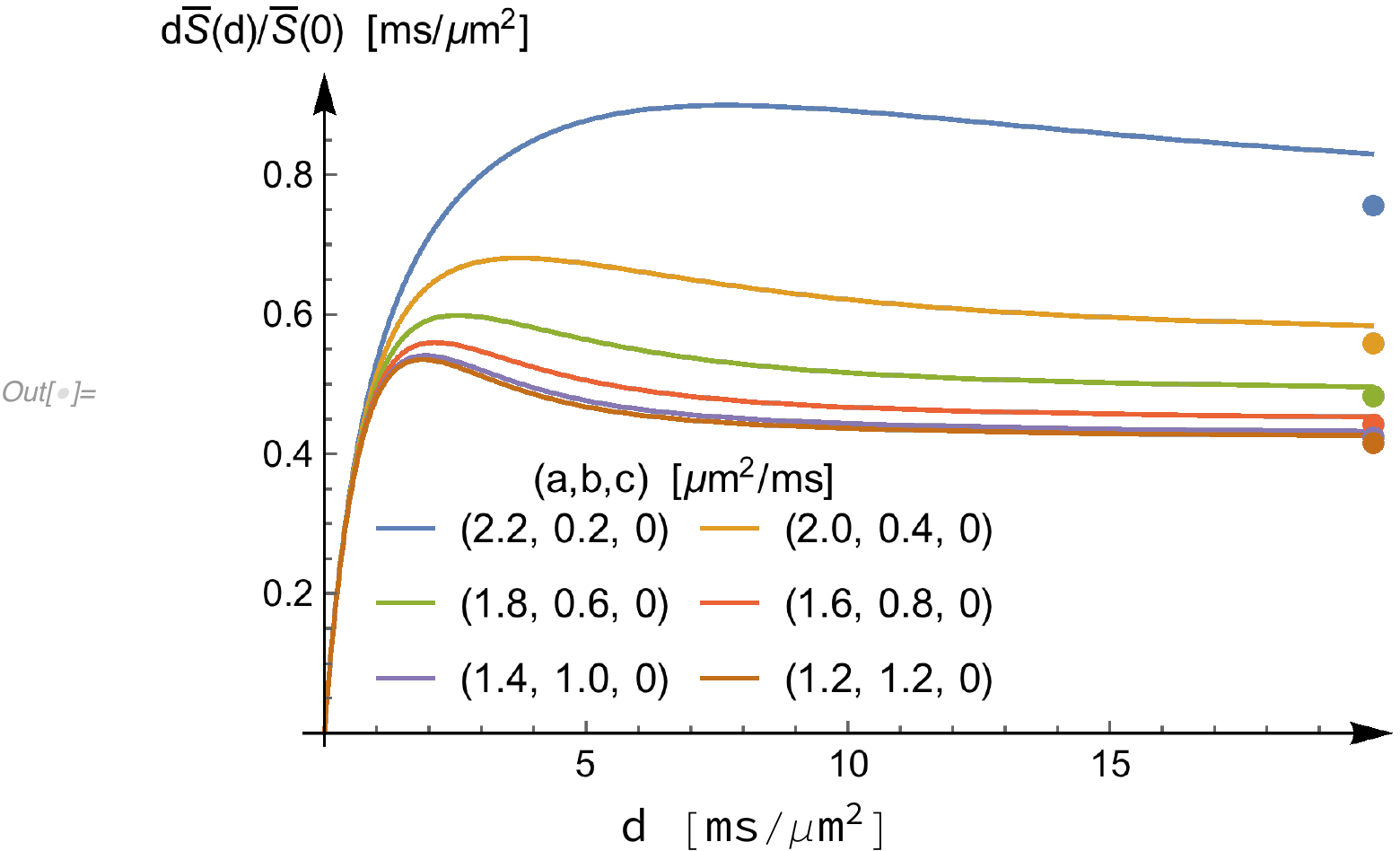}
\caption{
Given the rank-1 measurement tensor $\B$, with eigenvalues $(d,0,0)$, the product $d \bar S_{\A_i}(d)$ is plotted as a function of $d$ for six different diffusion tensors $\A_1, \A_2, \ldots \A_6$. These diffusion tensors all have the same trace, and their eigenvalues $(a,b,c)$ are in order
$\A_1:(2.2, 0.2, 0)$ $\mu$m${}^2/$ms,
$\A_2:(2.0, 0.4, 0)$ $\mu$m${}^2/$ms,
$\A_3:(1.8, 0.6, 0)$ $\mu$m${}^2/$ms,
$\A_4:(1.6, 0.8, 0)$ $\mu$m${}^2/$ms,
$\A_5:(1.4, 1.0, 0)$ $\mu$m${}^2/$ms,
$\A_6:(1.2, 1.2, 0)$ $\mu$m${}^2/$ms.
Using Theorem \ref{thm:kappa} $\lim_{d\to \infty} d \bar S_{\A_i}(d)$ are known, and the corresponding values are shown with dots.
The discrepancy between the curves and the corresponding dots are due to the finite values of $d$.
}
\label{fig:decay2}
\end{center}
\end{figure}
This limiting behaviour is illustrated in Figure \ref{fig:decay2} where, starting with a measurement tensor $ \B=d\begin{psmallmatrix} 1 &0 & 0 \\0 & 0 & 0 \\ 0 & 0 & 0 \end{psmallmatrix}$, plots of the signal $\bar S_{\A_i}(d)$ are shown for six different diffusion tensors $\A_i$, $i=1,2\ldots 6$.
The eigenvalues $(a,b,c)$ of these diffusion tensors are 
$\A_1:(2.2, .2, 0)$ $\mu$m${}^2/$ms,
$\A_2:(2.0, .4, 0)$ $\mu$m${}^2/$ms,
$\A_3:(1.8, .6, 0)$ $\mu$m${}^2/$ms,
$\A_4:(1.6, .8, 0)$ $\mu$m${}^2/$ms,
$\A_5:(1.4, 1.0, 0)$ $\mu$m${}^2/$ms,
$\A_6:(1.2, 1.2, 0)$ $\mu$m${}^2/$ms.
Note that all diffusion tensors have the same trace.
For each such tensor, the conditions of Theorem \ref{thm:kappa} are met, and the limits (as $d \to \infty$) $\tfrac{1}{2\sqrt{ab}}$  can be calculated.
These limits are shown with dots in Figure \ref{fig:decay2}. The discrepancy between the curves and the corresponding dots are due to the finite values of $d$.

Hence, for large $d$,  $\bar S_{a,b,0}(d) \propto \left[\kappa(a,b) \, d\right]^{-1}$ with $\kappa(a,b)=2\sqrt{a b}$,
which evaluates to $\Tr(\A)$ only for $a=b$.
 On the other hand, the small $d$ behaviour of the average signal is always $\bar S_{\A} \approx 1-\frac{d}{3}\Tr(\A)$. Thus, $\Tr(\A)$ can, in principle, be determined from the initial decay of the average signal. Any mismatch between this estimate and the estimate of $\kappa(a,b)$ from the power-law decay of the signal can thus be attributed to transverse anisotropy of $\A$. 

The above inference is a clear example to how the derived expressions involving general tensors can be utilized to gain new insight into the characterization of the local structure.

\subsection*{Estimation of $\A$ from the power series expansion of $\bar S$}\label{sec:estimationofD}

One of the main results of this paper consists of the series in
Eq.\ \eqref{eq:agenbgen}, which make it possible to calculate the
powder average $\bar S$ for general matrices $\A$ and $\B$ without
numerical integration or sampling in $\mathrm{SO}(3)$. As the
corresponding formula \eqref{eq:agenbaxi} for the case when one of the
matrices, $\B$, is axisymmetric is simpler, a relevant question is
whether the use of a general $\B$ is of any help in determining $\A$
from a set of measurements. This question immediately splits into two;
namely if there is an advantage with a general $\B$ in principle, and
also in practice. For instance, the asymptotics of $\bar S$ for
``large'' $\B$ may be of primary importance in principle, but in
practice the actual values of $\bar S$ in that regime may be so small
that noise and measurement errors make them impossible to use for determining $\A$.

To address the question above, we start by considering a general $\B$ and then the special case when
$\B$ is axisymmetric. Since any axisymmetric matrix $\B$ is the sum of an isotropic matrix and a rank-one matrix, and since the effect on $\bar S$ of an isotropic matrix is trivial, c.f.\ Eq.\ \eqref{eq:isotropic}, it is sufficient to consider rank-one matrices $\B$.

So, with a general $\A$, we consider a family of $\B$-matrices given by
$\B=\lambda \widetilde \B$, where 
$\widetilde \B=\begin{psmallmatrix}
d & 0 & 0 \\
0 & e & 0 \\
0 & 0 & f
\end{psmallmatrix}
$
is a fixed matrix, and $\lambda$ is a scalar strength parameter. For fixed $\A$ and $\widetilde \B$, the orientationally-averaged signal $\bar S_{\A}(\B)=\bar S_{\A}(\lambda\widetilde \B)=\bar S(\lambda)$ is a function of $\lambda$, and we seek the dependence for small $\lambda$.  Consider therefore the expansion $\bar S_{\A}(\lambda\widetilde \B)=1+c_1(\A,\widetilde \B)\lambda+c_2(\A,\widetilde \B)\lambda^2+c_3(\A,\widetilde \B)\lambda^3+ \mathcal{O}(\lambda^4)$, $ \lambda \to 0$. The coefficients $c_1, c_2, c_3$ can be estimated from the knowledge of $\{\bar S_{\A, \widetilde \B}(\lambda_i)\}_i$
for a suitable set of (small) $\{ \lambda_i \}_i$.
Up to third order (see \CYY{Supplementary Section III}) 
the coefficients $c_k$ are
\begin{equation}\label{eq:c123}
\begin{split}
c_1 =& -\frac{1}{3}\Tr(\A)\Tr(\widetilde \B)\\
c_2 =& \frac{1}{30}\left(2\Tr^2(\A)\Tr^2(\widetilde \B)+3\Tr(\A^2)\Tr(\widetilde \B^2)\right. \\ 
&\left. -\Tr^2(\A)\Tr(\widetilde \B^2)-\Tr(\A^2)\Tr^2(\widetilde \B)
\right)\\
c_3=&\frac{1}{630}\left[ \Tr(\A^3) \left(-36 \Tr(\widetilde \B^3)+36 \Tr(\widetilde \B^2) \Tr(\widetilde \B)-8\Tr^3(\widetilde \B)\right) \right.\\
&+\Tr(\A^2) \Tr(\A) \left(36 \Tr(\widetilde \B^3)-57 \Tr(\widetilde \B^2) \Tr(\widetilde \B)+15\Tr^3(\widetilde \B)\right)\\
&\left. +\Tr^3(\A) \left(-8 \Tr(\widetilde \B^3)+15 \Tr(\widetilde \B^2) \Tr(\widetilde \B)-8 \Tr^3(\widetilde \B)\right)\right].
   \end{split}
\end{equation}
For a generic matrix $\widetilde \B$ \CYY{(with the isotropic choice being the singular exception)}, knowledge of these three coefficients suffice to determine the matrix $\A$: For ease of illustration, consider the special case wherein $\widetilde \B$ is rank-1 (i.e., $e=f=0$) yielding the simplified coefficients
\begin{equation}\label{eq:Brank1}
\begin{split}
c_1 =& -\frac{d}{3}\Tr(\A)\\
c_2 =& \frac{d^2}{30}\left(\Tr^2(\A)+2\Tr(\A^2)\right)\\
c_3=&-\frac{d^3}{630}\left(\Tr^3(\A)+6\Tr(\A^2) \Tr(\A)+8 \Tr(\A^3)\right) \ .
\end{split}
\end{equation}
We see that this is a non-degenerate system that gives $\Tr(\A), \Tr(\A^2), \Tr(\A^3)$ in terms of $c_1, c_2, c_3$, and since (the eigenvalues of) $\A$ is determined from these three traces, $c_1, c_2$ and $c_3$ determine $\A$. 

Returning to a general $\widetilde \B$, a convenient point of view is to consider $\widetilde \B$ as decomposed into two parts (c.f., \CYY{Supplementary Section I}) 
and write, $\widetilde \B=\begin{psmallmatrix} \delta & 0 & 0
\\ 0 & \epsilon & 0 \\ 0 & 0 & 0
\end{psmallmatrix} + f \I$, where $\delta=d-f, \epsilon=e-f$.
Then the average signal \eqref{eq:S_avg} factors into two, each factor calculable from 
$\widetilde \B=\begin{psmallmatrix} \delta & 0 & 0 \\ 0 & \epsilon & 0 \\ 0 & 0 & 0 \end{psmallmatrix}$
 and $\widetilde \B = f \I$ separately.
We find that when
$\widetilde \B=\begin{psmallmatrix}
\delta & 0 & 0 \\
0 & \epsilon & 0 \\
0 & 0 & 0
\end{psmallmatrix}$, the expansion coefficients \eqref{eq:c123} attain the simplified form
\begin{equation}\label{eq:c123d}
\begin{split}
c_1 =& -\frac{1}{3}(\delta+\epsilon)\Tr(\A)\\ c_2 =&
\frac{1}{30}\left[(\delta^2+\epsilon^2)\left(\Tr^2(\A)+2\Tr(\A^2)\right)+\delta\epsilon \left(4
  \Tr^2(\A)-2\Tr(\A^2)\right)
  \right]\\ c_3=&-\frac{1}{630} \big[(\delta^3+\epsilon^3)\left(\Tr^3(\A)+6 \Tr(\A)
  \Tr(\A^2)+8 \Tr(\A^3)\right)\\ &+(\delta^2\epsilon+\delta\epsilon^2)\left(9 \Tr^3(\A)+12
  \Tr(\A)\Tr(\A^2) -12 \Tr(\A^3)\right) \big] \ ,
\end{split}
\end{equation}
while an isotropic $\widetilde \B = f \I$ on the other hand yields (renaming the coefficients)
\begin{equation}\label{eq:c123iso}
g_1 = -f\Tr(\A), \quad
g_2 = \frac{f^2}{2}\Tr^2(\A), \quad 
g_3=-\frac{f^3}{6}\Tr^3(\A) \ .
\end{equation}
The resulting signal expansion is the product $\bar{S} (\lambda) = [1
  +c_1 \lambda + c_2 \lambda^2+c_3 \lambda^3 +\mathcal{O}(\lambda^4)]
[1+g_1 \lambda + g_2 \lambda^2+g_3 \lambda^3
  +\mathcal{O}(\lambda^4)]$. Note how the isotropic part is not all that helpful: It attenuates the signal without being sensitive to anything other than $\Tr(\A)$. 

When using a general $\widetilde \B$ as described above, Eq.\ \eqref{eq:c123d} tells us that  $\A$ cannot be determined from $c_1$ and $c_2$ alone, even if  we are using the extra degree of freedom in $\B$ by varying $\delta$ and $\epsilon$. This is so since $c_1$ and $c_2$ contain information only on two invariants ($\Tr(\A)$ and $\Tr(\A^2)$), which is not sufficient to determine $\A$, which has three eigenvalues. That being said, in practice, varying $\delta$ and $\epsilon$ may be a way of getting ``independent'' measurements to reduce the influence of noise, i.e., increase SNR and therefore provide more reliable estimates of the coefficients $c_k$.

Let us also mention that in the situation described by Eq.\ \eqref{eq:c123d} where the combined freedom of $\delta$ and $\epsilon$ offers no extra information, one can still vary these to test the assumption that the specimen is indeed described by a single diffusivity matrix $\A$. Namely, for different values of $\delta$ and $\epsilon$, one should always get the same estimated matrix $\A$.

\subsection*{Standard model of white-matter with a general diffusion tensor for the extracellular compartment}

The extra degree of freedom provided by the parameters $\delta$ and $\epsilon$ may be crucial in certain relevant situations. In principle, such  situations can be addressed by involving higher order coefficients $c_4, c_5, \ldots$ but reliably estimating them would get increasingly difficult, and using $c_1, c_2, c_3$ for various $\delta$ and $\epsilon$ values is more robust. We give the following example, which could be relevant for simplified models of white-matter microstructure.

Suppose that our specimen is a mixture of two different substances, in unknown proportions, where one of the substances is (from a diffusivity perspective) a ``stick''. Thus, in proportions $p$ and $1-p$, with $0 \leq p \leq 1$, we have the unknown diffusivity matrices $\A \sim \begin{psmallmatrix}
a & 0 & 0 \\
0 & b & 0 \\
0 & 0 & c
\end{psmallmatrix}$ and $\widetilde \A \sim \begin{psmallmatrix}
q & 0 & 0 \\
0 & 0 & 0 \\
0 & 0 & 0
\end{psmallmatrix}$. This model is similar to the commonly employed white-matter model (sometimes referred to as the ``standard model'' \cite{Novikov18NBMreview}), but differs from it in two ways: (i) we ignore the isotropic compartment, whose contribution to the orientationally-averaged signal is simply an exponential (see \eqref{eq:isotropic}), (ii) the contribution from the extracellular matrix is given by a general diffusion tensor, which is not necessarily axially-symmetric.

In this model, we therefore have five unknowns to determine; three invariants of $\A$ together with $p$ and $q$. Since $\Tr(\widetilde \A^k)=q^k, k=1,2,3, \ldots$, Eq.\ \eqref{eq:c123d} now becomes
\begin{equation}\label{eq:c123dtwo}
\begin{split}
(\delta+\epsilon)\left(p\Tr(\A) +(1-p)q\right) & = -3 c_1(\delta,\epsilon) \\
(\delta^2+\epsilon^2) \left[ p\left(\Tr^2(\A)+2\Tr(\A^2)\right)+(1-p)3q^2 \right]\\
+\delta\epsilon \left[p \left(4 \Tr^2(\A)-2\Tr(\A^2) \right)+ (1-p)2q^2 \right]& = 30 c_2(\delta,\epsilon)\\
(\delta^3+\epsilon^3) \left[ p \left(\Tr^3(\A)+6 \Tr(\A) \Tr(\A^2)+8 \Tr(\A^3) \right)+(1-p)15q^3 \right]\\
+(\delta^2\epsilon+\delta\epsilon^2) \left[p \left( 9 \Tr^3(\A)+12 \Tr(\A)\Tr(\A^2) -12 \Tr(\A^3) \right)+(1-p)9q^3 \right]&=-630 c_3(\delta,\epsilon).
\end{split}
\end{equation}
By varying $\delta$ and/or $\epsilon$, these equations determine $p,q, \A$ ``almost
uniquely'', in the sense that there are (rare) situations where there
are two sets of acceptable solutions to Eq.\ \eqref{eq:c123dtwo}.
However, if one also adds the measurements with an isotropic
measurement tensor, this possible ambiguity goes away.  (See \CYY{Supplementary Section IV} 
for details.) Thus, it is possible to obtain all 5 unknowns of a multi-compartment white-matter model from the first three terms of the power-series representation of the orientationally-averaged signal.

\section*{Discussion and Conclusions}

\CYY{Our elaborations on the numerical behaviour of the main results \eqref{eq:agenbaxi} and \eqref{eq:agenbgen} were restricted to the case where one matrix is axisymmetric, corresponding to Eq.\ \eqref{eq:agenbaxi}. In the more general case corresponding to Eq. \eqref{eq:agenbgen}, the same guidelines apply regarding how the ordering of the eigenvalues affects the numerical behaviour, with some additional caveats. First of all, since no eigenvalues necessarily coincide, the number of possible distinguishable orderings increases, by a factor of 6 in particular. Also, not only one has to choose between the three alternatives (\ref{eq:agenbgen}c--e) by the same guidelines that apply to the alternatives (\ref{eq:agenbaxi}a--d), but also the number of coefficients $q_{mk}$ has to be specified.}

\CYY{A more brute-force approach to calculate the orientational average \eqref{eq:S_avg} is to set up the necessary integrations in the space of rotations and perform them by some discretization scheme. However, such an approach would require the integrand $e^{-\Tr (\A \R \B \R^\mT)}$ to be relatively well-behaved in the integration domain in order to be accurate, which is not always the case. For instance, when the matrices $\A$ and $\B$ are very similar to each other, and with one eigenvalue dominating the others, the integrand is (virtually) zero for most rotations $\R$, with most of the contribution stemming from a small subset. In such cases, a discretization of the average is not reliable.}


In this work, we represented local diffusion by employing a diffusion tensor along with the generalization \cite{Karlicek80} of the Stejskal-Tanner formula \cite{Stejskal65} for the signal contribution of each microdomain. The underlying assumption is that diffusion in each and every microdomain is unrestricted. This assumption \cite{Jian07} has been the building block of not only the microstructure models mentioned above, but also the techniques \cite{Westin14miccai,Lasic14,Szczepankiewicz15,Westin16_QTI} developed within the multidimensional diffusion MRI framework. Introduction of the confinement tensor concept \cite{Yolcu16,Ozarslan17FiP} provides a viable direction that could achieve the same by accounting for the possible restricted character of the microdomains. This is the subject of future work.

Another limitation of the present work is the assumption that there is no variation in the size of the microdomains making up the complex environment---the same assumption employed in many of the microstructure models. Previous studies \cite{Yablonskiy03,OzarslanNJP11} suggest the complexity of accounting for such variations in different contexts. We intend to address this issue in the future. 



In conclusion, we studied the orientationally-averaged magnetic resonance signal by extending the existing expressions to cases involving general tensors with no axisymmetry. This was accomplished by evaluating a challenging average (Eqs. \eqref{eq:S_avg}), or equivalently the integral in \eqref{eq:Laplace} for the special class of $p(\A)$ distributions considered in this article. Although the results are given as sums of infinitely many terms, we showed that with certain arrangements of the parameters, obtaining very accurate estimates is possible by retaining a few terms in the series. These developments led to a number of interesting inferences on the properties of the signal decay curve as well as estimation of relevant parameters from the signal. 

The findings presented in this work could be useful in many contexts in which the the expression  \eqref{eq:S_avg} (or \eqref{eq:Laplace}) emerges. For example, though we employed the nomenclature of diffusion MR in this paper, our findings are applicable to solid-state NMR spectroscopy as well due to the mathematical similarities of the two fields.  



\bibliography{../../sharedbib.bib}



\section*{Acknowledgements}

The authors thank Daniel Topgaard for his comments on the existing literature and acknowledge the following sources for funding: Swedish Foundation for Strategic Research AM13-0090, the Swedish Research Council CADICS Linneaus research environment, the Swedish Research Council 2015-05356 and 2016-04482, Link\"oping University Center for Industrial Information Technology (CENIIT), VINNOVA/ITEA3 17021 IMPACT, the Knut and Alice Wallenberg Foundation project``Seeing Organ Function,'' and National Institutes of Health P41EB015902.


\section*{Author contributions statement}

M.H.\ derived the main results. M.H., C.Y., and E.\"O.\ collaborated in application to specific cases and writing the manuscript, with inputs from H.K.\ and CF.W. All authors reviewed the manuscript.

\section*{Additional information}

\noindent \textbf{Supplementary information} accompanies this paper.

\noindent \textbf{Competing interests:} The authors declare no competing interests.




\end{document}



\title{{\bfseries Orientationally-averaged diffusion-attenuated magnetic resonance signal for locally-anisotropic diffusion: Supplementary information} }

\author{Magnus Herberthson}
\email[Electronic address: ]{magnus.herberthson@liu.se}
\affiliation{Dept.\ of Mathematics, Link\"oping University, Link\"oping, Sweden}
\author{Cem Yolcu}
\affiliation{Dept.\ of Biomedical Engineering, Link\"oping University, Link\"oping, Sweden}
\author{Hans Knutsson}
\affiliation{Dept.\ of Biomedical Engineering, Link\"oping University, Link\"oping, Sweden}
\author{Carl-Fredrik Westin}
\affiliation{Dept.\ of Biomedical Engineering, Link\"oping University, Link\"oping, Sweden}
\affiliation{Laboratory for Mathematics in Imaging, Dept.\ of Radiology, Brigham and Women's Hospital, Harvard Medical School, Boston, MA, USA}
\author{Evren \"Ozarslan}
\affiliation{Dept.\ of Biomedical Engineering, Link\"oping University, Link\"oping, Sweden}
\affiliation{Center for Medical Image Science and Visualization, Link\"oping University, Link\"oping, Sweden}





\maketitle 




\renewcommand{\theequation}{S\arabic{equation}}

\section{Calculation of the integrals} \label{appendix:derivation}

In this appendix, we derive the formulas (5--10). 
After some setup and general remarks, we start by deriving Eq.\ (8), 
i.e., the case when one matrix is axisymmetric, in \ref{appendix:agenbaxi}. From
this the formulas (5--7) 
follow as special cases.

The result (8) 
for the case involving one axisymmetric matrix is in
itself a special case of the general result (10) 
where both $\A$ and $\B$ are general. However, this situation is more
involved and is addressed in \ref{appendix:agenbgen}.

Note that as Eqs.\ (2) stand, it is enough that either $\A$ or $\B$ is
symmetric, as $\Tr(\R^\mT \A \R \B ) = 0$ if $\A$ is symmetric and $\B$ is
antisymmetric, or vice versa. Also, by the physical set up, it is
natural to consider all the eigenvalues of $\A$ and $\B$ to be non-negative,
but this is not insisted on in the calculations. Moreover, throughout
the article, fractional factorials are defined via the Gamma function,
i.e., $z! = \Gamma (z + 1)$.

We start by rewriting 
(in some basis which is not important)
\begin{align}
\A &=c \I+ \left(\begin{smallmatrix} 
\alpha & 0 & 0 \\
0 & \beta & 0 \\
0 & 0 & 0
\end{smallmatrix} \right), \, \text{with} \nonumber \\
\alpha &= a-c \ , \nonumber \\
\beta &= b-c \nonumber \ ,
\end{align}
and similarly (in perhaps a different basis)
\begin{align*}
\B &=
f\I+ \left(\begin{smallmatrix}
\delta & 0 & 0 \\
0 & \epsilon & 0 \\
0 & 0 & 0
\end{smallmatrix} \right), \, \text{with} \\
\delta &= d-f \ , \\
\epsilon &= e-f \ .
\end{align*}
Thus, $B=f\I+\delta \uu \uu^\mT+\epsilon \vv \vv^\mT$ for some pair of orthogonal unit vectors $\uu,\vv$. Using this, we find that
\begin{align}
  \Tr( \A \R \B \R^\mT) &=c(d+e)+f(a+b)-cf+
\Tr \!\left[ \left(\begin{smallmatrix}
\alpha & 0 & 0 \\
0 & \beta & 0 \\
0 & 0 & 0 
\end{smallmatrix} \right)
  \R (\delta \uu \uu^\mT+\epsilon \vv \vv^\mT) \R^\mT \right] \nonumber
\end{align}
Noting that $\Tr (\A {\bm x} {\bm x}^\mT)={\bm x}^\mT \A {\bm x}$, and defining
\begin{align}
  Q_0 = c(d+e) +f(a+b) -cf \ , \label{eq:Q0}
\end{align}
the original expression (2) 
for the powder-averaged signal becomes
\begin{equation} \label{eq:uvexp}
\bar{S} = e^{-Q_0} \left< e^{- \delta (\R \uu)^\mT \left( \begin{smallmatrix}
\alpha & 0 & 0 \\
0 & \beta & 0 \\
0 & 0 & 0 
\end{smallmatrix} \right) 
\R \uu -\epsilon (\R \vv)^\mT \left( \begin{smallmatrix}
\alpha & 0 & 0 \\
0 & \beta & 0 \\
0 & 0 & 0 
\end{smallmatrix} \right)
\R \vv}\right>_{\R\in \mathrm{SO(3)}} .
\end{equation}

The mean over rotation matrices $\R$ will be taken in the following
way (this is related to the Hopf fibration \cite{Hopf31,Lyons03} of $S^3$, with $S^3$ being the double cover of SO(3)). 

Given the
two orthogonal unit vectors $\uu$ and $\vv$, a rotation matrix $\R$ is
determined by the (mutually orthogonal) images $\R \uu$ and $\R \vv$.
Thus $\uu$ will run over all unit vectors and for each instance, $\vv$
will run over all possible directions orthogonal to it.  Using
standard spherical coordinates, we let
\begin{align}
  \uu =
\begin{psmallmatrix}
\sin\theta \cos\phi \\
\sin\theta \sin\phi \\
\cos\theta
\end{psmallmatrix}, \, 
\vv=\begin{psmallmatrix}
-\sin\phi \\
\cos\phi \\
0
\end{psmallmatrix},
\end{align}
and let $\R_\uu(\mu)$ denote a rotation matrix which rotates an angle $\mu$ around $\uu$.
From the exponent in Eq.\ \eqref{eq:uvexp}, we put 
\begin{equation}\label{eq:Q1}
\begin{split}
Q_1 &=\delta (\R_\uu(\mu) \uu)^\mT \begin{psmallmatrix}
\alpha & 0 & 0 \\
0 & \beta & 0 \\
0 & 0 & 0 
\end{psmallmatrix} 
\R_\uu(\mu) \uu=
\delta \uu^\mT \begin{psmallmatrix}
\alpha & 0 & 0 \\
0 & \beta & 0 \\
0 & 0 & 0 
\end{psmallmatrix} 
\uu \\
&=\delta \sin^2\theta(\alpha \cos^2\phi+\beta \sin^2\phi)=\delta (\alpha+(\beta-\alpha) \sin^2\phi)\sin^2\theta,
\end{split}
\end{equation}
which is independent of $\mu$, 
and
\begin{equation}\label{eq:Q2}
\begin{split}
Q_2= \epsilon \left(\R_\uu(\mu) \vv
\right)^\mT
\begin{psmallmatrix}
\alpha & 0 & 0 \\
0 & \beta & 0 \\
0 & 0 & 0 
\end{psmallmatrix} 
\R_\uu(\mu) 
\vv,
\end{split}
\end{equation}
where it easy to check that
$$
R_\uu(\mu) \vv=
R_\uu(\mu) \begin{psmallmatrix} -\sin\phi \\ \cos\phi \\ 0 \end{psmallmatrix}
=
\begin{psmallmatrix}
 -\sin \mu \cos \phi \cos \theta-\cos\mu  \sin \phi  \\
 \cos\mu \cos\phi-\sin\mu\sin\phi\cos\theta \\
 \sin\mu \sin\theta \\
\end{psmallmatrix} .
$$
Now the average \eqref{eq:uvexp} can be rewritten as, 
\begin{equation}\label{eq:tripint}
\begin{split}
\bar{S}
&=\frac{e^{-Q_0}}{(4\pi) (2\pi)} \int\limits_0^{2\pi}\int\limits_0^\pi
e^{-Q_1} \int\limits_0^{2\pi} e^{-Q_2} \, \dd\mu
\sin\theta \, \dd\theta \, \dd\phi \ .
\end{split}
\end{equation}
Evaluation of the angular average in the above expression in the
general case is provided in \ref{appendix:agenbgen}, while the
resulting expressions were given in \CYY{{\slshape Results}}.

However, we start by looking at the special case with one matrix, say
$\B$, being axisymmetric.  One can then arrange its eigenvalues so
that $\epsilon=0$, which yields $Q_2=0$ and means that the angular
average \eqref{eq:tripint} takes the form
\begin{align}\label{eq:tvaint}
\bar{S} &= \frac{e^{-Q_0}}{4\pi}\int \limits_0^{2\pi}\int\limits_0^\pi 
e^{-Q_1} \sin\theta \, \dd\theta \, \dd\phi \nonumber \\
&= \frac{e^{-Q_0}}{4\pi}\int \limits_0^{2\pi}\int\limits_0^\pi 
e^{-\delta (\alpha+(\beta-\alpha) \sin^2\phi)\sin^2\theta}\sin\theta \, \dd\theta \, \dd\phi\ ,
\end{align}
as a result of the $\mu$ integral dropping out since $Q_2$ is the
only place where $\mu$-dependence might have been present.  Moreover,
one of the integrations can readily be reduced by using either
$$
\frac{1}{2\pi}\int_0^{2\pi}e^{-a\sin^2\phi} \,\dd\phi=e^{-a/2}I_0(\tfrac{a}{2})
$$
or 
$$
\frac{1}{\pi}\int_0^{\pi}e^{-a\sin^2\theta}\sin\theta \, \dd\theta=
\frac{e^{-a}}{\sqrt{a\pi}} \text{erfi}\left(\sqrt{a}\right).
$$
Here $I_0$ is a modified Bessel function of the first kind (with order
$0$) and $\text{erfi}$ is the imaginary error function,
$\text{erfi}(x)=\text{erf}(i x)/i$. Hence, with one of the matrices
$\A$ and $\B$ axisymmetric, it is easy to write the orientational
average (2) 
as definite integral of one
variable. However, aiming for the result (8), 
where we
have expressed \eqref{eq:tvaint} as a series rather than an integral,
we proceed as follows.

\subsection{$\A$ general, $\B$ axisymmetric}\label{appendix:agenbaxi}

Suppose that one matrix, say $\B$, is axisymmetric, and that we let the two equal eigenvalues be $e=f$, so that $\epsilon=0$. As noted in above, 
$\bar S$ is then given by Eq.\ \eqref{eq:tvaint}.
Putting for brevity $\zeta=\delta\alpha$ and $\gamma=\delta(\beta-\alpha)$, the integrand has the form 
\begin{align*}
e^{-(\zeta+\gamma\sin^2\phi)\sin^2\theta} &=
\sum\limits_{n=0}^\infty\frac{(-\sin^2\theta)^n}{n!}(\zeta+\gamma\sin^2\phi)^n \\
&= \sum\limits_{n=0}^\infty\frac{(-\sin^2\theta)^n}{n!}\sum\limits_{k=0}^n
{n \choose k}\zeta^{n-k}\gamma^k\sin^{2k}\phi \ . 
\end{align*}
Now, it is easily verified that $\int\limits_{0}^{2\pi} \sin^{2k}\!\phi \, \dd\phi=2\pi4^{-k}{2k \choose k}$ and that
$\int\limits_{0}^{\pi} \sin^{2n+1}\!\theta \, \dd\theta=\frac{\sqrt{\pi}n!}{(n+1/2)!}$. Hence \eqref{eq:tvaint} becomes (assuming for the moment $\alpha\neq 0$)
\begin{equation}\label{eq:appsum1}
\bar{S} = e^{-Q_0} \frac{\sqrt{\pi}}{2}\sum\limits_{n=0}^\infty\frac{(-\alpha\delta)^n}{(n+1/2)!}
\sum\limits_{k=0}^n{n \choose k}{2k \choose k}\left(\frac{\beta-\alpha}{4\alpha}\right)^k .
\end{equation}
From the definition of the hypergeometric function $\, _2 F_1$ (see \ref{appendix:hyp1F1}),  it is 
also straightforward to check that the inner sum can be identified as,
$_2F_1 \!\left( \tfrac{1}{2},-n; 1; \tfrac{\alpha-\beta}{\alpha} \right)$.
Thereby we obtain
\begin{align}
  \bar{S} &= e^{-Q_0} \frac{\sqrt{\pi}}{2}\sum\limits_{n=0}^\infty\frac{(-\alpha\delta)^n}{(n+1/2)!}
      {}_2F_1 \!\left( \tfrac{1}{2},-n; 1; \tfrac{\alpha-\beta}{\alpha} \right)
      \end{align}
Substituting the value \eqref{eq:Q0} of $Q_0$, as well as $\alpha$,
$\beta$, and $\delta$, we obtain the first form (8a) 
when $\alpha\neq 0$, i.e., $a\neq c$. The case $\alpha=a-c=0$ can be
handled explicitly from the start, or through
continuity arguments using \eqref{eq:zlimit}, in the form $\lim_{x \to 0}(-x)^n {}_2
F_1(\tfrac{1}{2},-n;1;\tfrac{y}{x})=
\tfrac{(n-1/2)!}{\sqrt{\pi}n!}y^n$, which leads to (8a) 
when $a=c$ upon recognizing the Taylor series for
the error function.

The form (8b) 
is obtained upon changing the order of
summation in Eq.\ \eqref{eq:appsum1}. We obtain
\begin{equation*}
\bar{S} = e^{-Q_0} \frac{\sqrt{\pi}}{2}\sum\limits_{k=0}^\infty{2k \choose k} 
\left(\frac{\beta-\alpha}{4\alpha}\right)^k\sum\limits_{n=k}^\infty\frac{(-\alpha\delta)^n}{(n+1/2)!}{n \choose k} \ ,
\end{equation*}
and from the relation
$$
\sum\limits_{n=k}^\infty\frac{(-\alpha\delta)^n}{(n+1/2)!}{n \choose k}=
\frac{(-\alpha \delta )^k }{\left(k+1/2\right)!}\, _1F_1\left(k+1;k+\tfrac{3}{2};-\alpha\delta \right)
$$
the second form (8b) 
follows (after renaming the summation index $k\to n$). 

Another way the sum \eqref{eq:appsum1} can be rewritten is via the change of summation indices, $m=n-k$, which yields, 
\begin{equation*}
\bar{S} = e^{-Q_0} \frac{\sqrt{\pi}}{2}\sum\limits_{m=0}^\infty\sum\limits_{k=0}^\infty\frac{(-\alpha \delta)^{m}(-1)^k}{(m+k+1/2)!}
{m+k \choose k}{2k \choose k}\left[\frac{\delta(\beta-\alpha)}{4}\right]^k \ .
\end{equation*}
The identity
$$
\sum_{k=0}^{\infty }
\frac{\left[\tfrac {\delta(\alpha-\beta)}{4} \right]^k\binom{2 k}{k} \binom{k+m}{k}}{\left(k+m+\frac{1}{2}\right)!}=
\frac{\, _2F_2 \!\left(\frac{1}{2},m+1;1,m+\frac{3}{2};\delta(\alpha-\beta) \right)}{(m+ \tfrac 12)!}
$$
then gives the form (8c) 
similarly.

\subsubsection*{Special cases}

All expressions (8) 
for the powder-averaged signal of course give the same result, and $\bar S$ is also invariant under permutations of $a,b,c$. One can exploit this to derive the special cases where the matrix $\A$ is also axisymmetric, that is, it has two eigenvalues coinciding.

Swapping $a$ and $c$ thanks to the aforementioned invariance,
and then putting $b=c$ in Eq.\ (8a), 
we get 
$$
\bar S  = e^{-a d-2 f c} \frac{\sqrt{\pi}}{2} \sum_{n=0}^a \frac{[(a-c)(d-f)]^n }{\left(n+1/2\right)!},
$$
which corresponds to \CYY{case 3 of {\slshape Results}} when both $\A$ and $\B$ are axisymmetric.
The identity
$$
\sum_{n=0}^\infty\frac{x^n}{(n+1/2)!}=
\frac{e^x \text{erf}\left(\sqrt{x}\right)}{\sqrt{x}}=
\frac{e^{x} \text{erfi}\left(\sqrt{-x}\right)}{\sqrt{-x}}
$$
then yields the associated result (7). 

The case where $\B$ is not only axisymmetric but also rank-1 follows by putting $f=0$, which gives Eq.\ (6). 

The simplest case of isotropic $\B$ can be verified as follows: As $\epsilon=0$ already, isotropy implies $\delta=0$, and the expressions in (8) 
readily give Eq.\ (5), 
due to all but the $n=0$
term vanishing in the series expansions.

\subsection{$\A$ and $\B$ both general}\label{appendix:agenbgen}

The problem now is to calculate the integral in
Eq.\ \eqref{eq:tripint}, where we no longer assume that
$\epsilon=0$. As a result, the integrand of Eq.\ \eqref{eq:tripint}
becomes dependent on $\mu$ through $Q_2$, and the $\mu$ integration
does not drop, contrary to the previous case.

From \eqref{eq:Q2} and the expression for $R_\uu(\mu) \vv$, we find that
\begin{equation}\label{eq:Q2B}
\begin{split}
Q_2 &=\frac{\epsilon}{2}[\tau+\rho \cos(2\mu)+\sigma\sin(2\mu)], \mbox{ where}\\
\tau &= \alpha+\beta-\left[\alpha+(\beta-\alpha) \sin ^2\phi \right]\sin ^2\theta\\
\rho &=-\alpha+\beta+2(\alpha-\beta) \sin ^2\phi+\left[\alpha+(\beta-\alpha) \sin ^2\phi \right]\sin ^2\theta\\
\sigma &=2 (\alpha-\beta) \sin\phi \cos\phi \cos\theta
\end{split}
\end{equation}
Hence, integrating over $\mu$, we find
\begin{equation}\label{eq:muint}
\begin{split}
\frac{1}{2\pi}\int_0^{2\pi}e^{-Q_2} \, \dd\mu
=\frac{e^{-\frac{\epsilon}{2}\tau}}{2\pi}\int_0^{2\pi}e^{-\frac{\epsilon}{2}[\rho \cos(2\mu)+\sigma \sin(2\mu)]} \, \dd\mu\\
=\frac{e^{-\frac{\epsilon}{2}\tau}}{2\pi}\int_0^{2\pi}e^{-\frac{\epsilon}{2}\sqrt{\rho^2+\sigma^2} \cos(2\mu)} \,\dd\mu
=e^{-\frac{\epsilon}{2}\tau} I_0\! \left(\tfrac{\epsilon}{2}\sqrt{\rho^2+\sigma^2}\right)
\end{split}
\end{equation}
where we have used $\int_0^{2\pi}e^{\pm A\cos(2\mu)} \,\dd\mu=2\pi
I_0(A)$, $I_0$ being the zeroth order modified Bessel function of the
first kind.  Inserting \eqref{eq:muint} into \eqref{eq:tripint}, the
integral is now
$$
\bar{S} = \frac {e^{-Q_0}}{4\pi} \int\limits_0^{2\pi}\int\limits_0^\pi
e^{-\frac{\epsilon}{2}\tau-\delta (\alpha+(\beta-\alpha) \sin^2\phi)\sin^2\theta}
I_0 \!\left(\tfrac{\epsilon}2 \sqrt{\rho^2+\sigma^2} \right) \sin\theta \, \dd\theta \, \dd\phi.
$$
For the argument of the exponential function, we have
\begin{equation}\label{eq:exparg}
\begin{split}
-\tfrac{\epsilon}{2}\tau &-\delta (\alpha+(\beta-\alpha) \sin^2\phi)\sin^2\theta\\
&=-\tfrac{\epsilon}{2}(\alpha+\beta)-(\delta+\tfrac{\epsilon}{2}) (\alpha+(\beta-\alpha) \sin^2\phi)\sin^2\theta\\
&=-\tfrac{\epsilon}{2}(\alpha+\beta)-(A+B\sin^2\phi)\sin^2\theta
\end{split}
\end{equation}
where $A=(\delta-\frac{\epsilon}{2})\alpha=\frac{1}{2} (a-c) (2 d-e-f)$, $B=(\delta-\frac{\epsilon}{2})(\beta-\alpha)=\frac{1}{2} (b-a) (2 d-e- f)$. Note that $I_0 (\frac{\epsilon}2 \sqrt{\rho^2+\sigma^2}) = I_0(\frac{|\epsilon|}2 \sqrt{\rho^2+\sigma^2})$ since $I_0$ is even.

For the argument of $I_0$ in \eqref{eq:muint} we find that
\begin{equation*}
\begin{split}
\rho^2+\sigma^2=(\alpha-\beta)^2+2 (\beta-\alpha)\left(\alpha-(\alpha+\beta) \sin ^2\phi\right)\sin ^2\theta
+\left(\alpha+(\beta-\alpha)  \sin ^2\phi\right)^2\sin ^4\theta
\end{split}
\end{equation*}
which means that
\begin{equation*}
\begin{split}
&\tfrac{|\epsilon|}{2}\sqrt{\rho^2+\sigma^2}=
\sqrt{C+(D+E\sin^2\phi)\sin^2\theta+(F+G\sin^2\phi)^2\sin^4\theta},
\end{split}
\end{equation*}
where $C=\frac{\epsilon^2}{4}(\alpha-\beta)^2=\frac{(e-f)^2}{4}(a-b)^2$, 
$D=\frac{\epsilon^2}{2}(\beta-\alpha)\alpha=\frac{(e-f)^2}{2}(b-a)(a-c)$, 
$E=\frac{\epsilon^2}{2}(\alpha^2-\beta^2)=\frac{(e-f)^2}{2}(a-b)(a+b-2c)$, 
$F=\frac{\epsilon\alpha}{2}=\frac{(e-f)(a-c)}{2}$, and 
$G=\frac{\epsilon(\beta-\alpha)}{2}=\frac{(e-f)(b-a)}{2}$.
Note that $C\geq0$.

Extracting the $\phi$- and $\theta$-independent part of
\eqref{eq:exparg}, and putting
$$e^{-Q_0'}=e^{-Q_0} e^{-\frac{\epsilon}{2}(\alpha+\beta)}=e^{-\frac{1}{2} (a+b) (e+f)-c  d}$$
the integral boils down to
\begin{equation}\label{eq:Xint}
\begin{split}
\bar{S} &= \frac{e^{-Q_0'}}{4\pi}\int\limits_0^{2\pi}\int\limits_0^\pi e^{-(A+B\sin^2\phi)\sin^2\theta}
I_0(X)\sin\theta \, \dd\theta \,\dd\phi, \\
X &= \sqrt{C+(D+E\sin^2\phi)\sin^2\theta+(F+G\sin^2\phi)^2\sin^4\theta} \ .
\end{split}
\end{equation}
The idea is now to combine the expansion from \ref{appendix:agenbaxi}:
\begin{equation}\label{eq:expexp}
\begin{split}
e^{-(A+B\sin^2\phi)\sin^2\theta}
=\sum\limits_{n=0}^\infty\sum\limits_{l=0}^n\frac{(-1)^n}{n!}
{n \choose l}A^{n-l}B^l\sin^{2n}\theta\sin^{2l}\phi
\end{split}
\end{equation}
with a Taylor expansion of $I_0(X)$. Namely, with $x=\sin^2 \phi, y=\sin^2 \theta$
we seek the Taylor expansion
\begin{equation}\label{eq:I0rel}
I_0(\sqrt{C+(D+E x)y+(F+G x)^2y^2})\!=\!\sum\limits_{k=0}^\infty\sum\limits_{m=0}^k q_{mk}x^m y^k
\end{equation}
where the particular form of the argument implies that the summation over $m$ (for each fixed $k$) runs
from $0$ to $k$. So, for each pair $(m,k)$ in this series and for each $(n,l)$ in the series \eqref{eq:expexp}
we use that
\begin{equation*}
\begin{split}
\frac{1}{4\pi}\int\limits_0^{2\pi}\int\limits_0^\pi \sin^{2n}\theta\sin^{2l}\phi \, x^m y^k\sin \theta \,\dd\theta \,\dd\phi
&=\frac{1}{4\pi}\int\limits_0^{2\pi}\int\limits_0^\pi \sin^{2(n+k)+1}\theta\sin^{2(l+m)}\phi \,\dd\theta \,\dd\phi\\
&=\frac{1}{2}4^{-m-l}{2(l+m) \choose l+m}\frac{\sqrt{\pi}(n+k)!}{(n+k+1/2)!}.
\end{split}
\end{equation*}
This means that \eqref{eq:Xint} takes the form
\begin{align}\label{eq:S_Y}
  \bar{S}=\sum\limits_{k=0}^\infty\sum\limits_{m=0}^k q_{mk} Y_{mk} \ ,
\end{align}
where
\begin{align} \label{eq:Ymk}
Y_{mk} = \sum\limits_{n=0}^\infty\sum\limits_{l=0}^n\frac{(-1)^n}{n!}
{n \choose l}A^{n-l}B^l
\frac{4^{-m-l}}{2}{2(l+m) \choose l+m}\frac{\sqrt{\pi}(n+k)!}{(n+k+1/2)!} \ .
\end{align}
This expression is comparable to \eqref{eq:appsum1}, which was explicitly summed in various ways to give the different expressions for $\bar S$ in Eq.\ (8). 
Here, we can proceed in the same way to obtain the three formulas given in Eq.\ (10), 
each corresponding to an alternative form $Y_{mk}^{(i)}$. 

The first alternative (10c) 
follows from taking the sum
over $l$, showing that $Y_{mk}$ can be written
\begin{equation}\label{Umk}
Y_{mk}^{(1)} (A,B) = \frac{\sqrt{\pi}4^{-m}}{2}\binom{2 m}{m}
\sum_{n=0}^\infty\frac{(-A)^n  (k+n)! \,
   _2F_1\left(m+\frac{1}{2},-n;m+1;-\frac{B}{A}\right)}
   {n! (k+n+1/2)!} \ ,
\end{equation}
where we recall $A=\frac{1}{2}(a-c)(2d-e-f)$ and $B=\frac{1}{2}(b-a)(2d-e-f)$.
This extends to $A=0$ by continuity via $\lim_{x\to 0}x^n \,
_2F_1\left(m+\frac{1}{2},-n;m+1;-\frac{1}{x}\right)=
\tfrac{(m+1/2)_n}{(m+1)_n}$, where $(a)_n$ is the Pochhammer symbol (see \ref{appendix:hyp1F1}), 
and leads to the exceptional case in (10c). 

Alternatively, rearranging
$\sum\limits_{n=0}^\infty\sum\limits_{l=0}^n=\sum\limits_{l=0}^\infty\sum\limits_{n=l}^\infty$
and summing over $n$ yields
\begin{equation}\label{Vmk}
Y_{mk}^{(2)}(A,B)=\frac{\sqrt{\pi}4^{-m}}{2}
\sum_{l=0}^\infty\frac{(-B)^l (k+l)! \binom{2 (l+m)}{l+m} \,
   _1F_1\left(k+l+1;k+l+\frac{3}{2};-A\right)}{4^l l! \left(k+ l+1/2\right)!} \ ,
\end{equation}
which is Eq.\ (10d) 
after renaming the summation index and substituting the values of $A$ and $B$.

Finally, one can rearrange the double summation in Eq.\ \eqref{eq:Ymk} by putting $n=\sigma+l$, yielding the third alternative (10e) 
after summation over $l$ as
\begin{equation}\label{Wmk}
Y_{mk}^{(3)}(A,B)=\frac{\sqrt{\pi}4^{-m}}{2}\binom{2 m}{m}
\sum_{\sigma=0}^\infty\frac{(-A)^{\sigma } (k+\sigma )! \,
   _2F_2\left(m+\frac{1}{2},k+\sigma +1;m+1,k+\sigma +\frac{3}{2};-B\right)}{\sigma !
   \left(k+\sigma +1/2\right)!} \ ,
\end{equation}
upon substituting $A$, $B$, and renaming the remaining summation index.

The derivation of the the expressions for $\bar S$ given in
Eq.\ (10) 
are thus complete once the form (10b) 
of the coefficients $q_{mk}$ in
Eq.\ \eqref{eq:S_Y} is demonstrated. These coefficients inherit their
form from the relation \eqref{eq:I0rel}.  We start by recalling that
(we will only need $m\in \mathbf{Z}$)
\begin{equation}\label{eq:Im}
I_m(x)=\sum_{n=0}^\infty\frac{1}{n!(n+m)!}\left(\frac{x}{2}\right)^{2n+m}.
\end{equation}
Hence, $I_m$ is even (odd) when $m$ is even (odd).
Again using $\frac{1}{n!}=0$ when $n$ is a negative integer, it also follows that $I_m=I_{-m}$.
Moreover, $m=0$ gives $
I_0(x)=\sum_{n=0}^\infty\frac{1}{n!^2}\left(\frac{x}{2}\right)^{2n}.
$
Using this and the trinomial expansion, we get
\begin{equation}\label{I0arg}
\begin{split}
&I_0(\sqrt{C+(D+Ex)y+(F+Gx)^2 y^2})=\sum\limits_{n=0}^\infty\frac{1}{4^n(n!)^2}[C+(D+Ex)y+(F+Gx)^2y^2]^n\\
&=\sum\limits_{n=0}^\infty\frac{1}{4^n(n!)^2}\times \hspace*{-2mm} \sum\limits_{n_1+n_2+n_3=n}\hspace*{-1mm}
{n \choose n_1 n_2 n_3}C^{n_1}(D+Ex)^{n_2}y^{n_2}(F+Gx)^{2n_3}y^{2n_3}
\end{split}
\end{equation}
where ${n \choose n_1 n_2 n_3}=\frac{n!}{n_1! n_2! n_3!}$. The coefficient (including powers of $x$) of $y^k$
is given by 
$$
\begin{cases}
n_1+n_2+n_3=n\\
n_2+2n_3=k
\end{cases}
\mbox{, i.e., }
\begin{cases}
n_1=n-k+n_3\\
n_2=k-2 n_3
\end{cases}.
$$
Replacing $n_1$, this coefficient is
\begin{equation}\label{eq:Ikn3}
\begin{split}
&\sum\limits_{n=0}^\infty\frac{1}{4^n n!}\! \! \sum\limits_{n_3=0}^{k/2}\! \!
\frac{C^{n-k+n_3}}{(n\!-\! k\! +\! n_3)!n_2!n_3!}(D+Ex)^{n_2}(F+Gx)^{2n_3}\\
&=\sum\limits_{n_3=0}^{k/2}\left(\tfrac{\sqrt{C}}{2}\right)^{k-n_3}
\hspace*{-.3ex} I_{k-n_3}(\sqrt{C})
\frac{C^{-k+n_3}}{n_2!n_3!}(D+Ex)^{n_2}(F+Gx)^{2n_3}.
\end{split}
\end{equation}
where we have used \eqref{eq:Im}.
From this we want to extract the coefficient of $x^m$. Now,
\begin{equation}\label{eq:DEFGpower}
\begin{split}
&(D+Ex)^{n_2}(F+Gx)^{2n_3}=\sum\limits_{\eta_2=0}^{n_2}
{n_2 \choose \eta_2}D^{n_2-\eta_2}E^{\eta_2}x^{\eta_2}
\sum\limits_{\eta_3=0}^{2n_3}{2n_3 \choose \eta_3}F^{2n_3-\eta_3}G^{\eta_3}x^{\eta_3},
\end{split}
\end{equation}
and hence the coefficient of $x^m$ is obtained when $\eta_2+\eta_3=m$. So, assuming that $\tfrac{EF}{DG}$ is well defined, and inserting 
$\eta_3=m-\eta_2$, the coefficient of $x^m$ in \eqref{eq:DEFGpower}
is
\begin{equation}\label{eq:DEFGpower2}
\begin{split}
&\sum\limits_{\eta_2=0}^{n_2}\!\!{n_2 \choose \eta_2}D^{n_2-\eta_2}E^{\eta_2}
{2n_3 \choose m-\eta_2}F^{2n_3-m+\eta_2}G^{m-\eta_2}\\
&=D^{n_2}F^{2n_3-m}G^m\sum\limits_{\eta_2=0}^{k-2n_3}\!\!
{n_2 \choose \eta_2}{2n_3 \choose m-\eta_2}\left(\frac{EF}{DG}\right)^{\eta_2}\\
&=D^{k-2n_3}F^{2n_3-m}G^m\sum\limits_{\eta_2=0}^{k-2n_3}\!\!
{k-2n_3 \choose \eta_2}{2n_3 \choose m-\eta_2}\left(\frac{EF}{DG}\right)^{\eta_2},
\end{split}
\end{equation}
where in the last equation we have used $n_2=k-2n_3$.

From the definition of the regularized hypergeometric function (see \ref{appendix:hyp1F1}) ${}_2\widetilde F_1 (\cdot)$
it follows that\footnote{ Valid if $x\neq 0$, otherwise the sum is ${2 n_3 \choose m }$.}
\begin{equation}\label{eq:hypFreg}
\begin{split}
&\sum\limits_{\eta_2=0}^{k-2n_3}\!\!
{k-2n_3 \choose \eta_2}{2n_3 \choose m-\eta_2}x^{\eta_2}
=x^{m-2 n_3}\frac{(k-2n_3)!}{(k-m)!}{}_2 \tilde F_1(m-k,-2n_3; 1 + m - 2 n_3;x).
\end{split}
\end{equation}
Combining \eqref{eq:Ikn3} (with $n_2=k-n_3$), \eqref{eq:DEFGpower2} and \eqref{eq:hypFreg} (with $x=\tfrac{EF}{DG}$)
we find that 
\begin{equation*}
\begin{split}
&q_{mk}=\left(\frac{E}{D}\right)^m\!\! \left(\frac{D}{2\sqrt{C}}\right)^k\!\! \frac{1}{(k-m)!}
\sum\limits_{n_3=0}^{k/2}\! \!
\frac{I_{k-n_3}(\sqrt{C})}{n_3!}
\!\!\left(\!\!\frac{2\sqrt{C}G^2}{E^2}\!\!\right)^{n_3}\!\!\!\!\!\!{}_2 \tilde F_1 \!\left(m\!-\!k,-2n_3; 1\! +\! m\! - 2 n_3;\frac{EF}{DG}\right).
\end{split}
\end{equation*}
Using the expressions for $C,D,E,F,G$, this (with $n_3=j$) gives
precisely the expression (10b) 
(when $a\neq b$ and $a+b-2c=0$) up to
the fact that $\sqrt{C}$ contains a modulus:
$\sqrt{C}=\frac{1}{2}|e-f|\cdot |a-b|$.  However, the parity of
$I_{k-n_3}$ together with the powers of $\sqrt{C}$ shows that the
modulus may be removed and hence the desired form of $q_{mk}$ is
established when $\tfrac{EF}{DG}\neq 0$, which corresponds to $a \neq b$ and
$a+b-2c \neq 0$.

When $a=b$, this implies that $C=D=E=G=0$, and inserting this in \eqref{I0arg}, we readily get the expression 
(10b) 
for the case $a=b$. Moreover, by putting $a=b$ in the formulas 
(10c--e), 
they all become (possibly after a summation) 
$
\tfrac{k!}{(k+1/2)!}{}_1 F_1(1+k; \tfrac{3}{2}+k; \tfrac{1}{2}(c-a)(2d-e-f))
$
so that (8d) 
follows. The case $a+b-2c=0$ can be handled by a limit procedure or by using the footnote preceding
Eq.~\eqref{eq:hypFreg}. Either way, this establishes (10b) 
when $a+b-2c=0$.


\section{Some hypergeometric facts} \label{appendix:hyp1F1}

In this appendix we will give some elementary facts from the  extensive field of hypergeometric functions, which are used at various places. Hypergeometric functions are power series $\sum\limits_{n=0}^\infty a_n x^n$ where the quotient $a_{n+1}/a_n$ is a rational function of $n$. We will mostly use the functions (from which the general pattern should be clear)
\begin{equation}\label{eq:hypdef}
{}_1F_1 \left(a; b; x \right)=\sum_{n=0}^{\infty}\frac{(a)_n}{(b)_n }\frac{x^n}{n!}, \
{}_2F_1 \left(a, b; c; x \right)=\sum_{n=0}^{\infty}\frac{(a)_n (b)_n}{(c)_n }\frac{x^n}{n!},\
{}_2F_2 \left(a, b; c, d; x \right)=\sum_{n=0}^{\infty}\frac{(a)_n (b)_n}{(c)_n (d)_n}\frac{x^n}{n!},
\end{equation}
where $(a)_n$ is the Pochhammer symbol
$$
(a)_n=a(a+1)(a+2) \cdots (a+n-1), 
$$
which is $\frac{\Gamma(a+n)}{\Gamma(a)}$ if $a$ is not a non-positive integer. 

It is immediate that all hypergeometric functions
have value $1$ when $x=0$.
From the definition \eqref{eq:hypdef}, it is easy to see that
 ${}_2F_1\left(\tfrac{1}{2}, -n; 1; x \right)$ are polynomials for $n\in \mathbf{N}$, which are positive and decreasing for negative $x$.
It also follows that 
$ _1F_1\left(k+1;k+\frac{3}{2};x \right)$ and $_2F_2\left(\frac{1}{2},m+1;1,m+\frac{3}{2};x\right)$ are positive and
 increasing for $x\geq 0$ for all $k,m \in \mathbf{N}$. In addition,
there is a relation $_1F_1\left(k+1;k+\frac{3}{2};-x \right)=e^{-x}
 {}_1F_1\left(\tfrac{1}{2};k+\frac{3}{2};x \right)$ which shows that for any
 $k\in \mathbf{N}$, $_1F_1\left(k+1;k+\frac{3}{2};x \right)>0$ for all
 $x$. It also holds that
$\frac{\partial }{\partial x}\,
_1F_1\left(k+1;k+\frac{3}{2};x\right)=\frac{(k+1) \,
  _1F_1\left(k+2;k+\frac{5}{2};x\right)}{k+\frac{3}{2}}$, so the conclusion is that for any nonnegative integer $k$,
 $_1F_1\left(k+1;k+\frac{3}{2};x \right)$ is positive and also increasing on $\mathbf{R}$.
(Similar remarks hold for ${}_2F_2(1/2, m + 1; 1, m + 3/2;x)$.)

In sections \ref{appendix:agenbaxi} and \ref{appendix:agenbgen}, we also use the result (for non-negative integers $n,m$)
\begin{equation}\label{eq:zlimit}
\lim_{x \to 0}x^n {}_2F_1(m+\tfrac{1}{2},-n;m+1;\tfrac{y}{x})=\frac{(m+1/2)_n}{(m+1)_n}(-y)^n
\end{equation}
which is rather immediate from the definition, bearing in mind that 
${}_2F_1(m+\tfrac{1}{2},-n;m+1;\tfrac{y}{x})$ is a polynomial of degree $n$ in the variable $\tfrac{y}{x}$.
$m=0$ gives the value $\frac{(1/2)_n}{(1)_n}(-y)^n=\frac{\Gamma(n+1/2)}{\Gamma(1/2)n!}(-y)^n=
\tfrac{(n-1/2)!}{\sqrt{\pi}n!}(-y)^n$.
Also, we use the regularized hypergeometric function ${}_2 \tilde F_1(a,b; c;x)$ which is related to ${}_2 F_1(a,b; c;x)$ via
$$
{}_2 \tilde F_1(a,b; c;x)=\frac{{}_2 F_1(a,b; c;x)}{\Gamma(c)}
$$

In \CYY{the section titled ``{\slshape Numerical Behaviour}''} we saw
that for some parameter values (suitably arranged), the series (8d) 
using the function $_1F_1$ gave very efficient
summations. In addition, the computation of $_1F_1$ for the particular
parameter values we need can also be facilitated by the explicit
expressions or the recursion formula given below:
\begin{equation*}
\begin{split}
{}_1F_1 \left(1; \tfrac{3}{2}; x \right)&=\frac{\sqrt{\pi } e^x \text{erf}\left(\sqrt{x}\right)}{2 \sqrt{x}} \ , \\
{}_1F_1 \left(2; \tfrac{5}{2}; x \right)&=\frac{3 \sqrt{\pi } e^x (2 x-1) \text{erf}\left(\sqrt{x}\right)}{8 x^{3/2}}+\frac{3}{4 x} \ ,\\
{}_1F_1 \left(3; \tfrac{7}{2}; x \right)
&=\frac{15 \left(\sqrt{\pi } e^x \left(4 x^2-4 x+3\right) \text{erf}\left(\sqrt{x}\right)+2 \sqrt{x} (2 x-3)\right)}{64 x^{5/2}} \ , \\
{}_1F_1 \left(4; \tfrac{9}{2}; x \right)
&=\frac{35 (\sqrt{\pi }\! e^x \left(8 x^3-12 x^2+18 x-15\right) \text{erf}\left(\sqrt{x}\right)}{256 x^{7/2}}\\
&+\frac{+70 \sqrt{x} \left(4 x^2-8 x+15\right))}{256 x^{7/2}} \ , \\
\, _1F_1\left(k+2;k+\frac{5}{2};x\right) &=
\frac{\left(k+\frac{3}{2}\right) \left(k+\frac{1}{2}\right)}{(k+1) x} \,
   _1F_1\left(k;k+\frac{1}{2};x\right) \\
&-\frac{\left(k+\frac{3}{2}\right)
   \left(k+\frac{1}{2}-x\right) }{(k+1) x}\, _1F_1\left(k+1;k+\frac{3}{2};x\right) \ .
\end{split}
\end{equation*}


\section{Series expansion of $\langle e^{-\lambda \Tr(\R^\mT \A \R \widetilde \B )} \rangle_\R$}\label{appendix:seriesexpansion}
In this appendix we comment on the derivation of the coefficients
$c_1, c_2, c_3$ in the expansion given in Eq.\ (16). 
So, with
$\bar S_{\A,\widetilde \B}(\lambda)=\langle e^{-\lambda \Tr(\R^\mT \A \R \widetilde \B )} \rangle_\R$, we seek 
$\bar S_{\A, \widetilde \B}(\lambda)=1+c_1(\A,\widetilde \B)\lambda+c_2(\A,\widetilde \B)\lambda^2+c_3(\A,\widetilde \B)\lambda^3+ \mathcal{O}(\lambda^4), \lambda \to 0$.
We assume that $\A$ has eigenvalues $a,b,c$ and that $\widetilde \B$ has eigenvalues $d,e,f$.
From the definition of $\bar S$, it is clear that $c_k$ is a homogeneous polynomial in $a,b,c,d,e,f$ of degree $2k$ and that
for fixed $d,e,f$, $c_k$ is a homogeneous symmetric polynomial in $a,b,c$ of degree $k$ (and vice versa with the roles of
$a,b,c$ and $d,e,f$ changed). For instance, $c_1$ must be proportional to $(a+b+c)(d+e+f)=\Tr(\A)\Tr(\widetilde \B)$. 
To derive the coefficients $c_i$, one can either use Eq.\ (10) 
with the replacement 
$d \to \lambda d, e \to \lambda e, f \to \lambda f$, and then use the known expansions w.r.t. $\lambda$ of all involved functions. An alternative is to
start with Eq.\ \eqref{eq:tripint}. Reusing the notation and inserting $e^{-Q_0}=e^{cf-c(d+e)-f(a+b)}$, we therefore look at
$e^{cf-c(d+e)-f(a+b)-Q_1-Q_2}$ where we express $Q_1$ given by \eqref{eq:Q1} and $Q_2$ given by \eqref{eq:Q2B} in terms of $a,b,c,d,e,f$.
By the same replacement $d \to \lambda d, e \to \lambda e, f \to \lambda f$, it is then straightforward to to look at the series expansion w.r.t $\lambda$ of the exponential function first,
and then perform the integrals as indicated in Eq.\ \eqref{eq:tripint}. Either way it is straightforward, but a bit tedious, to get
\begin{equation*}
\begin{split}
3 c_1 =& -(a+b+c) (d+e+f)\\
30 c_2 =& \left(a^2+b^2+c^2\right) \left(3 d^2+2 d (e+f)+3 e^2+2 e f+3 f^2\right)\\
& +2 (a (b+c)+b c)   \left(d^2+4 d (e+f)+e^2+4 e f+f^2\right)\\
210 c_3 =& -2 \left(a^3+b^3+c^3\right) \left(d^3-5 d^2 (e+f)-d \left(5 e^2+8 e f+5 f^2\right)+e^3-5
   e^2 f-5 e f^2+f^3\right)\\
& -(a+b+c) \left(a^2+b^2+c^2\right) \left(3 d^3+13(e+f)(d^2+e+f)+
   d\left(13 (e^2+f^2)+18 e f\right)+3 (e^3+f^3)\right)\\
& -2 a b c
   \left(d^3+9 d^2 (e+f)+d \left(9 e^2+48 e f+9 f^2\right)+e^3+9 e^2 f+9 e
   f^2+f^3\right)
\end{split}
\end{equation*}
and then compare with Eq.\ (16) 
to see that the expressions coincide.


\section{Solutions to equation (20)}\label{appendix:quasiuniqueness} 
In this appendix, we address \CYY{in more detail our claim that there are situations where using a general measurement tensor $\B$ is crucial in determining the diffusive properties of the specimen.} Namely, suppose we are given a specimen with two substances with diffusivity matrices $\A \sim \begin{psmallmatrix}
a & 0 & 0 \\
0 & b & 0 \\
0 & 0 & c
\end{psmallmatrix}$ and $\widetilde \A \sim \begin{psmallmatrix}
q & 0 & 0 \\
0 & 0 & 0 \\
0 & 0 & 0
\end{psmallmatrix}$ which occur in unknown proportions $p$ and $1-p$. We define $x=\Tr(\A)$, $y=\Tr(\A^2)$, $z=\Tr(\A^3)$. By the procedure described in \CYY{the section titled ``{\slshape Estimation of $\A$ from the power series expansion of $\bar{S}$}"} and using  the measurement tensor $\widetilde \B=\begin{psmallmatrix}
\delta & 0 & 0 \\
0 & \epsilon & 0 \\
0 & 0 & 0
\end{psmallmatrix}$,  we know the values of
\begin{subequations} \label{eq:measvalB}
\begin{align}
(\delta+\epsilon)(p x &+(1-p)q) \ , \\
(\delta^2+\epsilon^2)(p\left(x^2+2y\right)+(1-p)3q^2)&+\delta\epsilon(p(4 x^2-2y)+(1-p)2q^2) \ , \\
(\delta^3+\epsilon^3)\left(p(x^3+6 x y+8 z)+(1-p)15q^3 \right)
&+(\delta^2\epsilon+\delta\epsilon^2)(p(9 x^3+12 x y -12 z)+(1-p)9q^3) \ ,
\end{align}
\end{subequations}
for different values of $\delta$ and $\epsilon$. Using an isotropic measurement tensor, we also know the values of
\begin{subequations} \label{eq:measvalIso}
\begin{align}
p x &+(1-p)q \ , \\
p x^2 &+(1-p)q^2 \ , \\
p x^3 &+(1-p)q^3 \ .
\end{align}
\end{subequations}
By varying $\delta$ and $\epsilon$, we see that (\ref{eq:measvalB}b) and (\ref{eq:measvalB}c) contain two independent 
expressions each, and hence (\ref{eq:measvalB}a-c) will give five expressions. We also see that the quantity (\ref{eq:measvalIso}a)
is a special case of (\ref{eq:measvalB}a).  Furthermore, since the expression in (\ref{eq:measvalB}b) for $\delta=0$ subtracted from the same expression for  $\delta=\epsilon$ is 
$=5\epsilon^2(p x^2+(1-p)q^2)$, the quantity in (\ref{eq:measvalIso}b) is also redundant. In total, we get six equations:
\begin{subequations} \label{eq:measvalall}
\begin{align}
p x +(1-p)q &=k_1 \ , \\
p\left(x^2+2y\right)+(1-p)3q^2 &=k_2 \ , \\
p\left(x^2-y\right) &=k_3 \ , \\
p(x^3+6 x y+8 z)+(1-p)15q^3 &=k_4 \ , \\
p(x^3+x y -2 z) &=k_5 \ , \\
p x^3 +(1-p)q^3 &=m_1 ,
\end{align}
\end{subequations}
which follow from \eqref{eq:measvalB} and (\ref{eq:measvalIso}c) as follows.
Writing (L\ref{eq:measvalall}a) for the left hand side of equation (\ref{eq:measvalall}a) et cetera, we readily find that
(L\ref{eq:measvalall}a)=(\ref{eq:measvalB}a)$/(\delta+\epsilon)$, (L\ref{eq:measvalall}b)=$\tfrac{1}{\epsilon^2}$(\ref{eq:measvalB}b)$|{}_{\delta=0}$, 
(L\ref{eq:measvalall}c)=$\tfrac{3}{10\epsilon^2}$(\ref{eq:measvalB}b)$|{}_{\delta=\epsilon}$ - $\tfrac{4}{5\epsilon^2}$(\ref{eq:measvalB}b)$|{}_{\delta=0}$,
(L\ref{eq:measvalall}d)=$\tfrac{1}{\epsilon^3}$(\ref{eq:measvalB}c)$|{}_{\delta=0}$,
(L\ref{eq:measvalall}e)=$\tfrac{5}{84\epsilon^3}$(\ref{eq:measvalB}c)$|{}_{\delta=\epsilon}$ - $\tfrac{4}{21\epsilon^3}$(\ref{eq:measvalB}c)$|{}_{\delta=0}$
and (L\ref{eq:measvalall}f)=(\ref{eq:measvalIso}c).
Thus, the first five equations above come from the measurement tensor $\widetilde \B=\begin{psmallmatrix}
\delta & 0 & 0 \\
0 & \epsilon & 0 \\
0 & 0 & 0
\end{psmallmatrix}$ while (\ref{eq:measvalall}f) comes from isotropic measurements.
To recollect, by varying $\epsilon$ and $\delta$ in the measurement tensor $\widetilde \B=\begin{psmallmatrix}
\delta & 0 & 0 \\
0 & \epsilon & 0 \\
0 & 0 & 0
\end{psmallmatrix}$, and by using isotropic measurements, we can determine the left hand sides in equation \eqref{eq:measvalall}, which are quantities denoted
 $k_1, \ldots k_5, m_1$. From the setup of the problem, we know that $0\leq p \leq 1$ and $q,a,b,c \geq 0$. 

Given the values $k_1, \ldots k_5, m_1$, to what extent do these determine $p, q, a, b, c$?
It generically holds: The set of equations (\ref{eq:measvalall}a)--(\ref{eq:measvalall}e) has two sets of solutions $\{p,q,x,y,z\}$
and $\{P,Q,X,Y,Z\}$ while the inclusion of (\ref{eq:measvalall}f) makes the solution set unique. As an imprecise statement, it also holds that in ``most'' cases, the second solution set $\{P,Q,X,Y,Z\}$ is ruled out out as unphysical. This can be due to that
$P$ is outside the interval $[0,1]$, some of the quantities $Q,X,Y,Z$ are negative, or that the eigenvalues belonging to
a matrix $\hat \A$ with $X=\Tr(\hat \A)$, $Y=\Tr(\hat \A^2)$, $Z=\Tr(\hat \A^3)$ are complex. However, there are values of $p,q,a,b,c$ where the second solution $\{P,Q,X,Y,Z\}$ corresponds to a realistic specimen mixture.

We find the generic solutions to Eq.\ (\ref{eq:measvalall}a)--(\ref{eq:measvalall}e) as follows. (\ref{eq:measvalall}c) gives
$y= x^2-\frac{k_3}{p}$ and inserted in (\ref{eq:measvalall}e) we get
$z=x^3 -\frac{k_3 x+k_5}{2 p}$. Next, (\ref{eq:measvalall}a) gives $q= \frac{p x-k_1}{p-1}$ and inserting these values in (\ref{eq:measvalall}b), we find that
$p=\frac{3 k_1^2-k_2-2 k_3}{6 k_1 x-k_2-2 k_3-3 x^2}$. Finally, with these values inserted  in 
(\ref{eq:measvalall}d), we get an equation for $x$, namely
$$
15\left(3 k_1^2-k_2\right)x^2-3  (5 k_1 (k_2+4
   k_3)-(k_4+4 k_5))x-3 k_1 (k_4+4 k_5)+5
   (k_2+2 k_3)^2=0.
$$
Generically, this equation has two solutions $x, X$, which are both real (since one is known to be real),
and hence two solutions sets are produced. (There are a number of special cases that need attention; e.g., $p=0,1$,
$3k_1^2-k_2-2k_3=0$, $3k_1^2-k_2=0$. For instance, $p$ is obviously  non-unique when $\A=\widetilde \A$. Each such case is straightforward, but also questionable from a measurement point of view. ) As mentioned above, most often the ``second'' solution $\{P,Q,X,Y,Z\}$ can be ruled out as unphysical. On the other hand, starting with $p=9/10, q=5, a=1/10, b=1, c=3$ (so that $x=41/10$, $y=1001/100$, $z=28001/1000$)
then also $\{P,Q,X,Y,Z\}=$ $\{178409449/178940890$, $2467/270$, $55769/13357$, $ 10075315571/892047245$, $4243907217264727/119150750514650\}$ is a ``physical'' solution corresponding to a matrix with eigenvalues $\approx \{.181$, $.715$, $3.279\}$. However, invoking the measurement with the isotropic matrix, i.e., Eq.\ (\ref{eq:measvalall}f), $x$ is determined uniquely
since one easily finds, using (\ref{eq:measvalall}c)--(\ref{eq:measvalall}f) that $x=\frac{15 m_1-k_4-4 k_5}{10 k_3}$.

\section{Proof of Theorem 1} \label{appendix:kappa}
We use Eq.\ (8b) 
with $c=e=f=0$. Then $\bar S=\bar S_{a,b,0}(d)=
\sum \limits_{n=0}^\infty 
  \frac {\left[(a-b)d\right]^n} {\left(2n+1\right) n!} 
  {}_1F_1 \!\left(n+1;n+\tfrac{3}{2};-a d\right)$. The case $a=b$ is already accounted for, and without loss of generality, we can assume that $a>b$.
From \ref{appendix:hyp1F1} we use  $_1F_1\left(n+1;n+\frac{3}{2};-x \right)=e^{-x} {}_1F_1\left(\tfrac{1}{2};n+\frac{3}{2};x \right)$.
Writing $(a-b)d=(1-\frac{b}{a})ad=(1-\alpha)x$ with $x=ad$ and 
$0<\alpha=b/a<1$, the statement of Theorem 1 reads
\begin{equation}\label{eq:kappasum}
\lim_{x\to \infty}\frac{1}{a}\sum_{n=0}^{\infty}\frac{(1-\alpha)^n x^{n+1}}{(2n+1)n!}e^{-x} {}_1F_1\left(\frac{1}{2};n+\frac{3}{2};x \right)=\frac{1}{2\sqrt{ab}}.
\end{equation}
From Ref.\ \onlinecite{Abramowitzbook}, we find that $_1F_1\left(\tfrac{1}{2};n+\frac{3}{2};x \right)=\frac{(n+\frac{1}{2})!}{\sqrt{\pi}}e^xx^{-n-1}(1+\mathcal{O}(\frac{1}{x}))$
as $x \to \infty$. If we can justify changing the limit and the summation in 
Eq.\ (\ref{eq:kappasum}), we get
$$
\frac{1}{a}\sum_{n=0}^{\infty}\lim_{x \to \infty}\frac{(1-\alpha)^n x^{n+1}}{(2n+1)n!}e^{-x} \frac{(n+\frac{1}{2})!}{\sqrt{\pi}}e^xx^{-n-1} (1+\mathcal{O}(\tfrac{1}{x}))=
\frac{1}{a}\sum_{n=0}^\infty\frac{(1-\alpha)^n}{(2n+1)n!} \frac{(n+\frac{1}{2})!}{\sqrt{\pi}}=
\frac{1}{a}\frac{1}{2\sqrt{\alpha}}=\frac{1}{2\sqrt{ab}}.
$$
Here we have used that (for $|t|<1$) $\frac{1}{\sqrt{1-t}}=\sum\limits_{n=0}^\infty(-1)^n\binom{-1/2}{n}t^n$ with $t=1-\alpha$ and observed\footnote{
For this observation, the reflection formula $\Gamma(z)\Gamma(1-z)=\pi/\sin(\pi z)$ for non integer $z$ is helpful.
}
 that
$\frac{(-1)^n}{2}\binom{-1/2}{n}$ indeed equals $\frac{(n+\frac{1}{2})!}{\sqrt{\pi}(2n+1)n!}$ for $n=0,1,2,\ldots$. Hence Theorem 1 is proven if we can change summation and limit in Eq.\ (\ref{eq:kappasum}). To see that this is allowed, we define the functions
$f_n(x)=\frac{x^{n+1}}{(2n+1)n!}e^{-x} {}_1F_1\left(\tfrac{1}{2};n+\frac{3}{2};x \right)$ and consider\footnote{The factor $1/a$ is irrelevant.} 
$\sum\limits_{n=0}^\infty (1-\alpha)^n f_n(x)$. It is enough to prove that this series converges uniformly for $x\geq 0$, and this will be the case
if $\sum\limits_{n=0}^\infty (1-\alpha)^n M_n $ converges, where $M_n=||f_n||_\infty =\sup_{x \geq 0}|f_n(x)|$. 

To study $M_n$ we use \cite{Abramowitzbook} the integral representation
$
{}_1F_1\left(\tfrac{1}{2};n+\frac{3}{2};x \right)=\frac{\Gamma(n+\frac{3}{2})}{\sqrt{\pi}n!}\int_0^1e^{xt}\frac{(1-t)^n}{\sqrt{t}}\dd t,
$
so that (after a slight simplification)
$$
f_n(x)=\frac{(n-\frac{1}{2})!}{2\sqrt{\pi}n!}\frac{x^{n+1}e^{-x}}{n!}\int_0^1\frac{e^{xt}}{\sqrt{t}}(1-t)^n \dd t
\ ,$$
where we also put $c_n=\frac{(n-\frac{1}{2})!}{2\sqrt{\pi}n!}$ and
$g_n(x)=\frac{x^{n+1}e^{-x}}{n!}\int\limits_0^1\frac{e^{xt}}{\sqrt{t}}(1-t)^n \dd t$. Next, we write
\begin{align*}
  g_n(x)=\frac{x^{n+1}e^{-x}}{n!}\int\limits_0^1\frac{e^{xt}}{\sqrt{t}}(1-t)^n \dd t = g_{1n}(x)+g_{2n}(x) 
  =\frac{x^{n+1}e^{-x}}{n!}\int\limits_0^1\frac{1}{\sqrt{t}}(1-t)^n \dd t+\frac{x^{n+1}e^{-x}}{n!}\int\limits_0^1\frac{e^{xt}-1}{\sqrt{t}}(1-t)^n \dd t.
\end{align*}
Since $\int\limits_0^1\frac{1}{\sqrt{t}}(1-t)^n \dd t=\frac{\sqrt{\pi}n!}{(n+1/2)!}$, we find that $g_{1n}(x)=\frac{\sqrt{\pi}x^{n+1}e^{-x}}{(n+1/2)!}$, which attains its
maximum (for $x\geq 0$) at $x=n+1$ and hence $|| g_{1n}||_\infty = \frac{\sqrt{\pi}(n+1)^{n+1}e^{-(n+1)}}{(n+1/2)!}$. To estimate $g_{2n}$, we need an inequality.
First, with $\phi(v)=(1+v)(1-e^{-v})-2v$ it is easy to see that $\phi(v)\leq 0$ for $v \geq 0$. This means that (for $v \geq 0$) 
$1-e^{-v} \leq 2\frac{v}{1+v}\leq 2\sqrt\frac{v}{1+v}$ and finally $\frac{e^v-1}{\sqrt{v}}\leq \frac{2e^v}{\sqrt{1+v}}$. $g_{2n}$ can now be estimated through
$$
g_{2n}(x)=\frac{x^{n+1}e^{-x}}{n!}\int\limits_0^1\sqrt{x}\frac{e^{xt}-1}{\sqrt{xt}}(1-t)^n \dd t\leq
\frac{x^{n+1}e^{-x}}{n!}\int\limits_0^1\sqrt{x}\frac{2e^{xt}}{\sqrt{1+xt}}(1-t)^n \dd t.
$$
Splitting the integral into two, we first get
\begin{align*}
\frac{x^{n+1}e^{-x}}{n!}\int\limits_{1/2}^1\frac{2\sqrt{x}}{\sqrt{1+xt}}e^{xt}(1-t)^n \dd t &\leq \frac{x^{n+1}e^{-x}}{n!}\int\limits_{1/2}^12\sqrt{2}e^{xt}(1-t)^n \dd t \\ 
&=\frac{2\sqrt{2}x^{n+1}}{n!}\int\limits_{1/2}^1e^{x(t-1)}(1-t)^ndt \\ 
&= 
\frac{2\sqrt{2}x^{n+1}}{n!}\int\limits_{0}^{1/2}e^{-x s}s^n \dd s \\ 
& \leq \frac{2\sqrt{2}x^{n+1}}{n!}\int\limits_{0}^{\infty}e^{-x s}s^n \dd s \\
&=\frac{2\sqrt{2}x^{n+1}}{n!}\frac{n!}{x^{n+1}}=2\sqrt{2}. 
\end{align*}

For the remaining part, we get 
\begin{align*}
\frac{x^{n+1}e^{-x}}{n!}\int\limits_{0}^{1/2}\frac{2\sqrt{x}}{\sqrt{1+xt}}e^{xt}(1-t)^n \dd t &\leq
\frac{2x^{n+1}e^{-x}}{n!}\int\limits_{0}^{1/2}\sqrt{x}e^{xt}(1-t)^n \dd t \\
& = \frac{2x^{n+3/2}}{n!}\int\limits_{0}^{1/2}e^{x(t-1)}(1-t)^n \dd t \nonumber \\
& 
=\frac{2\sqrt{2}}{n!}\int\limits_{x/2}^{x}\sqrt{\frac{x}{2}}e^{-v}v^n \dd t \nonumber \\
& \leq \frac{2\sqrt{2}}{n!}\int\limits_{x/2}^{x}e^{-v}v^{n+1/2} \dd t \nonumber \\ 
& \leq \frac{2\sqrt{2}}{n!}\int\limits_{0}^{\infty}e^{-v}v^{n+1/2} \dd t \nonumber \\ 
& = \frac{2\sqrt{2}(n+1/2)!}{n!} \ . \nonumber
\end{align*}
Collecting terms, we see that
$$
M_n \leq \frac{(n-\frac{1}{2})!}{2\sqrt{\pi}n!} \left(
\frac{\sqrt{\pi}(n+1)^{n+1}e^{-(n+1)}}{(n+1/2)!}+2\sqrt{2}+\frac{2\sqrt{2}(n+1/2)!}{n!}
\right).
$$
By a straightforward use of Stirling's formula in the form\footnote{Here, $z$ need not be an integer.} $z!=\Gamma(z+1)=\sqrt{2\pi z}\left(\frac{z}{e}\right)^{z}(1+\mathcal{O}(\frac{1}{z}))$ as $z \to \infty$, we see that $\frac{(n-\frac{1}{2})!}{2\sqrt{\pi}n!} =\mathcal{O}(\frac{1}{\sqrt{n}})$, 
$\frac{\sqrt{\pi}(n+1)^{n+1}e^{-(n+1)}}{(n+1/2)!}=\mathcal{O}(1)$ and
$\frac{2\sqrt{2}(n+1/2)!}{n!}=\mathcal{O}(\sqrt{n})$, all as $n \to \infty$.
In total, $M_n=\mathcal{O}(1)$ as $n \to \infty$ so that $\sum\limits_{n=0}^\infty (1-\alpha)^n M_n $ converges (since $0<\alpha<1$). This concludes the proof.


\bibliography{../../../sharedbib.bib}